\documentclass[11pt]{article}
\usepackage{sustyle}
\usepackage{fullpage}
\usepackage{amsmath}
\usepackage{graphicx}
\usepackage{enumerate}
\usepackage{diagbox}
\usepackage{wrapfig}
\usepackage{tikz}
\usetikzlibrary{positioning}
\usepackage[round]{natbib}
\usepackage{url}  
\usepackage{multirow}
\usepackage{microtype}
\usepackage{amsmath,amsfonts,bbm,bbding,bm}
\usepackage{amsthm}
\usepackage{etoolbox}
\usepackage{subfigure}
\usepackage{booktabs} 
\usepackage{nicefrac}   
\usepackage{color} 
\usepackage{enumitem}
\usepackage{dsfont}
\usepackage[linktocpage]{hyperref}
\usepackage{lmodern}
\usepackage[T1]{fontenc}
\usepackage{authblk}
\author[1]{Zhekun Shi$^*$}
\author[1]{Kaizhao Liu\thanks{Co-first authors (order determined by a random coin flip). Emails: \texttt{\{zhekunshi,mrzt\}@stu.pku.edu.cn.}}} 
\author[2]{Qi Long\thanks{Co-corresponding authors. Emails: \texttt{qlong@upenn.edu}, \texttt{suw@wharton.upenn.edu}, \texttt{jcxiao@upenn.edu}.}}
\author[2]{Weijie J. Su$^\dagger$}
\author[2]{Jiancong Xiao$^\dagger$}
\affil[1]{Peking University}
\affil[2]{University of Pennsylvania}
\definecolor{blue}{rgb}{0, 0, 1}

\hypersetup{
colorlinks,
citecolor=gray,
linkcolor=blue,
urlcolor=black
}

\begin{document}

\title{Fundamental Limits of Game-Theoretic LLM Alignment: Smith Consistency and Preference Matching}
\date{\today}
\maketitle

\begin{abstract}

Nash Learning from Human Feedback (\texttt{NLHF}) is a game-theoretic framework for aligning large language models (LLMs) with human preferences by modeling learning as a two-player zero-sum game. 
However, using raw preference as the payoff in the game highly limits the potential of the game-theoretic LLM alignment framework.
In this paper, we systematically study using what choices of payoff based on the pairwise human preferences can yield desirable alignment properties. 
We establish necessary and sufficient conditions for Condorcet consistency, diversity through mixed strategies, and Smith consistency. 
These results provide a theoretical foundation for the robustness of game-theoretic LLM alignment.
Further, we show the impossibility of preference matching---i.e., no smooth and learnable mappings of pairwise preferences can guarantee a unique Nash equilibrium that matches a target policy, even under standard assumptions like the Bradley-Terry-Luce model. 
This result highlights the fundamental limitation of game-theoretic LLM alignment.

\end{abstract}

\section{Introduction}\label{sec:intro}

Large language models (LLMs) such as OpenAI-o3 \citep{openaio3} and DeepSeek-R1 \citep{DeepSeekAI2025DeepSeekR1IR} have demonstrated impressive capabilities across a wide range of domains, including code generation, data analysis, elementary mathematics, and reasoning \citep{hurst2024gpt,anthropic2024claude,chowdhery2023palm,touvron2023llama,ji2025overview}. These models are increasingly being used to tackle previously unsolved mathematical problems, drive scientific and algorithmic discoveries, optimize complex codebases, and support decision-making processes that were once considered unlikely to be automated in the near future \citep{bubeck2023sparks,eloundou2024gpts,alphaevolve}.

A key factor behind the popularity and effectiveness of LLMs is alignment: the process by which models learn to interact with human users and accommodate diverse human opinions and values by aligning their outputs with human preferences \citep{christiano2017deep}. 
The traditional method for alignment, reinforcement learning from human feedback (\texttt{RLHF}) \citep{Ouyang2022,casper2023open,dong2024rlhf}, typically begins by training a reward model on preference data collected from human labelers, often using the Bradley-Terry-Luce (BTL) model \citep{bradley1952rankanalysis,luce2012individual},
\begin{equation}
\label{eq:BT}
\cP(y \succ y'\mid x)=\frac{\exp(r(x,y))}{\exp(r(x,y))+\exp(r(x,y'))},
\end{equation}
where $r(x, y)$ is the reward function and $\cP(y\succ y'|x)$ is pairwise human preference, i.e., the fraction of individuals who prefer $y$ over $y'$ under prompt $x$.
In this framework, a higher scalar score assigned by the reward model to an LLM-generated response indicates a stronger preference by human labelers. 
The LLM is then fine-tuned through maximizing the reward to produce responses that are more likely to align with these preferences. 
However, \citet{munos2024nash} pointed out that the reward model cannot deal with preferences with cycles, and proposes an alternative alignment approach called Nash learning from human feedback (\texttt{NLHF}).
Unlike the reward-based methods, \texttt{NLHF} directly uses preference data to train a preference model and formulates LLM finetuning as finding Nash equilibrium in a two-player zero-sum game, also known as a von Neumann game \citep{myerson2013game}. 
Specifically, for a given prompt $x$, the LLM's policy $\bm{\pi}$ competes against an opposing policy $\bm{\pi}^\prime$ in a pairwise preference contest, where the objective is to find a policy that maximizes its worst-case preference score. Formally, \texttt{NLHF} solves the following min-max optimization problem:
\begin{equation*} \label{eq:ne} 
\max_{\bm{\pi}} \min_{\bm{\pi}^\prime} \mathbb{E}_{x\sim\rho}\left[\mathbb{E}_{y\sim\bm{\pi}(\cdot\mid x),y^\prime\sim\bm{\pi}^\prime(\cdot\mid x)}\left[\cP\left(y\succ y'\mid x\right)\right]\right],
\end{equation*}
where $\rho$ is a given distribution over prompts. 
However, \citet{munos2024nash} did not demonstrate the advantages of using the preference as the payoff in the game.

Recently, criteria from both social choice theory \citep{conitzer2024social,dai2024mapping,Mishra2023AIAA} and principles related to diversity \citep{xiao2024algorithmic,chakraborty2024maxmin} have been increasingly employed to scrutinize the alignment of LLM with human preference.
Notably, \texttt{RLHF} has been shown to fail both social choice theory considerations \citep{noothigattu2020axiom,siththaranjan2024distributional,ge2024axioms,liu2025statistical} and diversity considerations \citep{xiao2024algorithmic,chakraborty2024maxmin}.
In contrast, \texttt{NLHF} has been proved to enjoy these desirable properties.
It is shown in \citet{maurarivero2025jackpot} and \citet{liu2025statistical} that \texttt{NLHF} is \textit{Condorcet consistent} (see Definition \ref{def: Condorcet consistency}), meaning that the method always outputs the Condorcet winning response, a response that beats every other alternative response in pairwise majority comparisons, whenever one exists. Further, under a no-tie assumption (see Assumption \ref{ass:strict_preference}), \citet{liu2025statistical} showed that \texttt{NLHF} is \textit{Smith consistent} (see Definition \ref{def: Smith consistency}), meaning that the method always outputs responses from the Smith set, the smallest nonempty set of responses that pairwise dominate all alternatives outside the set. 
Moreover, \citet{liu2025statistical} showed that when human preference is diverse, i.e., there does not exist a single response that beat every other alternative, \texttt{NLHF} avoids collapsing to a single response by adopting a \textit{mixed strategy}.

Despite these advantages, there is no reason why we must use raw preference to design the payoff in the game-theoretic alignment approach. 
Using alternative payoffs in the game-theoretic LLM alignment framework might also lead to desirable alignment.
In this work, we systematically investigate the fundamental limits of the game-theoretic LLM alignment framework by analyzing how various choices of payoff influence its ability to satisfy key alignment criteria. We consider the following general game-theoretic alignment problem, involving a mapping applied to the preference denoted by $\Psi$:

\begin{equation}\label{eq:gen-nlhf}
\max_{\bm{\pi}} \min_{\bm{\pi}^{\prime}} \mathbb{E}_{x\sim\rho} \left[\mathbb{E}_{y\sim \bm{\pi}(\cdot\mid x)}\mathbb{E}_{y^\prime\sim \bm{\pi}^\prime(\cdot\mid x)}\left[\Psi\left(\mathcal{P}(y\succ y^{\prime}\mid x)\right)\right]\right].
\end{equation}
The general problem \eqref{eq:gen-nlhf} encompasses a range of games.
When $\Psi(t)=t$ is the identity mapping, the objective in Equation \eqref{eq:gen-nlhf} is equivalent to the standard \texttt{NLHF} objective. When $\Psi(t)=\log(t/(1-t))$ and the preference is generated by a BTL model, Equation \eqref{eq:gen-nlhf} recovers the standard \texttt{RLHF} objective. More importantly, the preference model $\cP_\theta$ used in practice is an estimation of the true human preference, which can be regarded as a noisy mapping of the ground-truth preference $\cP$. Allowing $\Psi$ to be stochastic provides a way to account for the uncertainty and noise inherent in estimating human preferences.
It is worth noting that a similar formalism has been proposed for non game-theoretic approach in \citet{azar2024general}, which used an non-decreasing mapping to process the preference.

In this paper, we first discuss when the solution to problem \eqref{eq:gen-nlhf} is Condorcet consistent and Smith consistent in Section \ref{sec:condorcet} and \ref{sec:smith} respectively. Our results show that these desirable properties are insensitive to the exact value of the payoff, revealing the robustness of game-theoretic alignment approaches. As a special case, we discover a natural generalization of \texttt{RLHF} objective that satisfy all these desirable properties. Technically, we develop novel proof techniques that can tackle a general non-symmetric game directly, instead of relying crucially on the symmetric nature of \texttt{NLHF} as in \citet{liu2025statistical}.

In addition, we examine the diversity of the solution by investigating whether the model produces a mixed strategy \citep{liu2025statistical}, and whether its output can satisfy the criterion of \textit{preference matching} \citep{xiao2024algorithmic}, meaning that the model output exactly matches a target policy which fully accounts for the diversity of human preference. 
Our findings suggest diversity can be ensured by mixed strategies, but exactly matching a target is difficult for any game-theoretic alignment approach.
This reveals a fundamental limitation of game-theoretic alignment approaches.

\subsection{Summary of Contributions}
We summarize our contributions as follows:
\begin{itemize}
    \item We show that Condorcet consistency is insensitive to the exact value of the payoff (Theorem \ref{thm:condorcet-consistent}), revealing the robustness of game-theoretic alignment approaches.
    \item We show that Smith consistency can be ensured by further maintaining the symmetry of the game (Theorem \ref{thm:smith}). Moreover, Smith consistent methods automatically preserve the diversity in human preferences by adopting mixed strategies (Corollary \ref{cor:smith diverse}).
    \item We show that preserving the diversity in human preference strictly, in the sense of preference matching, is impossible in general (Theorem \ref{thm: no pm}). This reveals a fundamental limitation of game-theoretic alignment approaches.
\end{itemize}

Assuming $\Psi$ is continuous, Table \ref{table: summarization} provides a concise summary of our mathematical results.

\begin{table}[htbp]
\caption{Summary of our mathematical results: the necessary and sufficient conditions on $\Psi$ to guarantee certain desirable alignment properties.}
\label{table: summarization}
\centering
\scalebox{0.87}{\begin{tabular}{c | c }
\toprule
Condorcet consistency & $\Psi(t)\geqslant \Psi(1/2)\,,\forall \, 1/2\leqslant t\leqslant 1$ and $\Psi(t)<\Psi(1/2)\,,\forall\, 0\leqslant t< 1/2$ \\ \midrule
Mixed \& Condorcet consistency & $\Psi(t)+\Psi(1-t)\geqslant 2\Psi(1/2)\,,\forall \, 1/2\leqslant t\leqslant 1$ and $\Psi(t)<\Psi(1/2)\,,\forall \, 0\leqslant t< 1/2$\\ \midrule
Smith consistency & $\Psi(t)+\Psi(1-t)= 2\Psi(1/2)\,,\forall \, 1/2\leqslant t\leqslant 1$ and $\Psi(t)<\Psi(1/2)\,,\forall \, 0\leqslant t< 1/2$\\ 
\bottomrule
\end{tabular}}
\end{table}

\subsection{Related Works} 

A general mapping $\Psi$ was first introduced in \citet{azar2024general} to facilitate the analysis of traditional non game-theoretic LLM alignment methodologies. Their objective function, called \texttt{$\Psi$PO}, applies a general mapping $\Psi$ to the original human preference.
In this way, they were able to treat \texttt{RLHF} and \texttt{DPO} as special cases of \texttt{$\Psi$PO} under BTL model and argue that these methods are prone to overfitting. To avoid overfitting, they took $\Psi$ to be identity and arrived at a new efficient algorithm called \texttt{IPO}.
Our problem \eqref{eq:gen-nlhf} can be regarded as the analogy of \texttt{$\Psi$PO} in the context of game-theoretic LLM alignment.
Another difference is that rather than focusing on statistical properties like overfitting, our focus is on the alignment properties such as Smith consistency and preference matching.
Moreover, they restricted $\Psi$ to be a non-decreasing map, while we allow $\Psi$ to be arbitrary, even stochastic.

Condorcet consistency is one of the dominant concept in the theory of voting \citep{gehrlein2006condorcet,alma990019554260106761}, and Smith consistency is its natural generalization \citep{shoham2008multiagent,borgers2010mathematics}.
They have not been studied in the context of LLM alignment until recently \citep{maurarivero2025jackpot,liu2025statistical}.
In \citet{maurarivero2025jackpot}, the authors showed that \texttt{NLHF} with a selection probability that deals with ties is Condorcet consistent. Under a no-tie assumption, \citet{liu2025statistical} showed that \texttt{NLHF} is Condorcet consistent and Smith consistent, whereas \texttt{RLHF} is not unless the preference satisfies a BTL model. Further, the paper showed that the probability that the preference satisfies a BTL model is vanishing under the impartial culture assumption, highlighting a key advantage of the \texttt{NLHF} framework. 

Several recent works also focus on aligning LLMs with the diverse human preference \citep{chakraborty2024maxmin,xiao2024algorithmic,liu2025statistical}.
In \citet{chakraborty2024maxmin}, the authors introduced a mixture model to account for the opinion of minority group and arrive at the \texttt{MaxMin-RLHF} method. In \citet{xiao2024algorithmic}, the authors introduced the concept of preference matching and develop the \texttt{PM-RLHF} objective to pursue this goal. 
\citet{liu2025statistical} demonstrated that the original \texttt{NLHF} yields a mixed strategy when no Condorcet winning response exists, whereas standard \texttt{RLHF} produces a deterministic strategy, highlighting a potential advantage of \texttt{NLHF} in preserving the diversity of human preferences.

\section{Preliminaries}
Consider a general mapping $\Psi: [0,1]\to\mathbb{R}$. We apply $\Psi$ to the preference and study the max-min problem \eqref{eq:gen-nlhf} with this generalized payoff.
Any solution $\bm{\pi}$ employed by the first player at the Nash equilibrium,

\begin{equation}\label{eq:gen-nash-solution}
    \bm{\pi}\in\arg\max_{\bm{\pi}} \min_{\bm{\pi}^{\prime}} \mathbb{E}_{x\sim\rho} \left[\mathbb{E}_{y\sim \bm{\pi}(\cdot\mid x)}\mathbb{E}_{y^\prime\sim \bm{\pi}^\prime(\cdot\mid x)}\left[\Psi\left(\mathcal{P}(y\succ y^{\prime}\mid x)\right)\right]\right]\,,
\end{equation}
is called a Nash solution to the problem \eqref{eq:gen-nlhf}. The Nash solution is the policy which fully aligned LLMs will perform. 
Note that the set of Nash solutions remain the same after an overall shift of payoff, that is, changing $\Psi$ to $\Psi+C$ for any constant $C$ will not affect the problem.
The original \texttt{NLHF} objective \citep{munos2024nash} corresponds to the special case where $\Psi(t) = t$, equivalent to $\Psi(t)=t-1/2$, and the resulting game is symmetric \citep{duersch2012pure}, meaning that the two players are the same. However, for an arbitrary mapping $\Psi$, the game is usually not symmetric, and we only focus on the Nash solution employed by the first player.

Given a prompt $x$, we consider the set of all possible responses generated by the LLM: $\{y_1, \ldots, y_n\}$, where $n$ is the total number of possible responses. 
Without any loss of generality, we drop the dependence on the prompt $x$ from now on. For any two distinct response $y$ and $y'$, recall that $\mathcal{P}(y \succ y^\prime)$ denote the preference of $y$ over $y'$, defined as the expected proportion of individuals who prefer $y$ over $y^\prime$. 
By definition, human preference satisfies the condition $\mathcal{P}(y\succ y^{\prime})+\mathcal{P}(y^{\prime}\succ y)=1$ and naturally we let $\mathcal{P}(y\succ y) = 1/2$ \citep{munos2024nash}.
For any distinct pair of responses $y$ and $y^{\prime}$, we say that $y$ beats $y^{\prime}$ if $\mathcal{P}(y \succ y^{\prime}) > 1/2$. 
Additionally, following \citet{liu2025statistical}, we adopt the No-Tie assumption throughout this paper.

\begin{assumption}[No-Tie]
\label{ass:strict_preference}
For any distinct responses \( y \) and \( y' \), we assume that \( \mathcal{P} (y \succ y') \neq 1/2 \). 
\end{assumption}  

This assumption is both minimal and practically reasonable. First, if the number of labelers is odd, it automatically holds. Even in cases where a tie occurs, it can always be resolved through a more precise comparison.

\paragraph{Notation.} 
For any set $A$, we denote its cardinality by $\vert A\vert$. For any $n \in \mathbb{N}_+$, we define $[n] := \{1, \ldots, n\}$. We use $\delta_{ij} := \mathds{1}\{i = j\}$ for $1 \leqslant i, j \leqslant n$.
We represent high-dimensional vectors using bold symbols. Any policy $\bm{\pi}$ over the set of possible responses $\{y_1, \ldots, y_n\}$ can be identified with a vector in $\mathbb{R}^n$, where each entry $\pi_i$ corresponds to the probability assigned to $y_i$ for $i \in [n]$.
We then define the support of a policy $\bm{\pi}$ as $\operatorname{supp}(\bm{\pi}) := \{y_i \mid \pi_i > 0\,, i \in [n]\}$. We write $\bm{\pi} > 0$ if $\pi_i > 0$ for all $i \in [n]$, and similarly, $\bm{\pi} \geqslant 0$ if $\pi_i \geqslant 0$ for all $i \in [n]$.
For simplicity, we denote $\Psi(\mathcal{P}(y_i \succ y_j))$ as $\Psi_{ij}$ for any $1 \leqslant i, j \leqslant n$, and define the payoff matrix as $\boldsymbol{\Psi}:= \{\Psi_{ij}\}_{1 \leqslant i, j \leqslant n}$.
We then define the total payoff by:
\[
\mathcal{P}_{\Psi}(\boldsymbol{\pi}_1, \boldsymbol{\pi}_2) := \sum_{i=1}^n \sum_{j=1}^n \pi_{1,i} \pi_{2,j} \Psi_{ij}\,.
\]
We denote by $\bm{\delta}_i$ the policy supported solely on $y_i$, i.e., $\operatorname{supp}(\bm{\delta}_i) = \{y_i\}$. Furthermore, the mixed policy $(\bm{\delta}_{i_1} + \ldots + \bm{\delta}_{i_k})/k$ is defined as the policy $\bm{\pi}$ such that
\[
\pi_i = \begin{cases}
1/k\,, &  i \in \{i_1, \cdots, i_k\}\\
0\,, & \text{otherwise}
\end{cases}\,,
\]
for any subset $\{i_1, \ldots, i_k\} \subseteq [n]$.
\section{Condorcet Consistency}\label{sec:condorcet}

In this section, we examine Condorcet consistency—a desirable property for LLM alignment inspired by social choice theory—within the generalized game-theoretic LLM fine-tuning framework \eqref{eq:gen-nlhf}. We begin by defining the Condorcet winning response and Condorcet consistency. We then present Theorem \ref{thm:condorcet-consistent}, which characterizes the necessary and sufficient conditions on the mapping $\Psi$ to guarantee Condorcet consistency. Next, we examine the conditions under which $\Psi$ preserves human preference diversity when no Condorcet winner exists and introduce Theorem \ref{thm:diversity}. Finally, we discuss the continuity assumption underlying Theorem \ref{thm:diversity}. We provide the proofs of Theorems \ref{thm:condorcet-consistent} and \ref{thm:diversity} in Section \ref{sec:proof-condorcet} and Section \ref{sec:proof-diversity}, respectively.

Following \citet{liu2025statistical}, a response that is preferred over all others in pairwise comparisons by the preference model is referred to as the Condorcet winning response. 

\begin{definition}[Condorcet Winning Response]\label{def:Condorcet winning response}  
A response \( y^\star \) is called a Condorcet winning response if \( \mathcal{P}(y^\star \succ y) > 1/2 \) for all \( y \neq y^\star \).  
\end{definition}

It is clear that there can be at most one Condorcet winning response. When such a response exists, a natural requirement for LLM alignment is that this response should be the output. This property is known as Condorcet consistency.

\begin{definition}[Condorcet Consistency]\label{def: Condorcet consistency}
    Problem \eqref{eq:gen-nlhf} is Condorcet consistent if 
    when there exists a Condorcet winning response, the Nash solution to \eqref{eq:gen-nlhf} is unique and corresponds to this Condorcet winning response.
\end{definition}

\citet{liu2025statistical} and \citet{maurarivero2025jackpot} showed that the original \texttt{NLHF} objective, which corresponds to the case where $\Psi(\cdot)$ is identity, is Condorcet consistent. In this paper, we proceed further and investigate the following question:

\begin{center}
\emph{ Which choices of $\Psi$ ensure Condorcet consistency?}
\end{center}

We answer this question in Theorem \ref{thm:condorcet-consistent} and the proof is deferred to Section \ref{sec:proof-condorcet}.

\begin{theorem}\label{thm:condorcet-consistent}
Problem \eqref{eq:gen-nlhf} is Condorcet consistent if and only if $\Psi(\cdot)$ satisfies
\begin{equation}\label{eq:thm-condorcet-consistent}
\begin{cases}
\Psi(t)\geqslant\Psi(1/2)\,,1\geqslant t\geqslant 1/2\\ \Psi(t)<\Psi(1/2)\,, 1/2>t\geqslant 0
\end{cases}\,.
\end{equation}
\end{theorem}

Note that this condition is much weaker than requiring $\Psi$ to be increasing. It only demands that $\Psi$ maps any value greater than 1/2 to some value larger than $\Psi(1/2)$, and any value less than $1/2$ to some value smaller than $\Psi(1/2)$. This implies that a wide range of mapping functions can be used within the game-theoretic LLM alignment framework \eqref{eq:gen-nlhf} to ensure Condorcet consistency. 
Furthermore, in practice, we do not have access to the ground-truth preference model. 
Instead, we parameterize the preference model using a deep neural network, $\mathcal{P}_\theta(y \succ y^\prime)$, trained on large-scale preference datasets \citep{munos2024nash}. Due to the limitations of the datasets and the optimization process, the learned model only approximates the true human preferences. We can view this approximation as $\Psi(\mathcal{P}(y \succ y^\prime))$ in our framework. In practice, we can enforce the parameterized preference model to satisfy $\mathcal{P}_\theta(y \succ y) = 1/2$, then our results show that as long as this approximation yields the correct pairwise majority comparisons—specifically, that $\mathcal{P}_\theta(y\succ y^\prime)\geqslant 1/2> \mathcal{P}_\theta(y^\prime\succ y)$ whenever $y$ beats $y^\prime$—then the LLM alignment remains Condorcet consistent. This strongly highlights the robustness of the game-theoretic LLM alignment approach in achieving Condorcet consistency.

When a Condorcet winning response does not exist, human preferences are diverse and there is no single response that is better than others. Therefore, in order to preserve the diversity inherent in human preferences, it is natural to require the Nash solution not to collapse to a single response. This motivation leads to the following characterization of diversity through mixed strategies.

\begin{definition}[Mixed Strategies]\label{def:mixed strategy}
    A Nash solution $\bm{\pi}$ is called a mixed strategy if $|\operatorname{supp}(\bm{\pi})|>1$.
\end{definition}

\citet{liu2025statistical} demonstrated that the original \texttt{NLHF}, which corresponds to the case where $\Psi(\cdot)$ is identity, yields a mixed strategy when no Condorcet winning response exists. Assuming that problem \eqref{eq:gen-nlhf} is Condorcet consistent, we proceed further and investigate: 
\begin{center}
    \emph{ Which choices of $\Psi$ lead to a mixed strategy in the absence of a Condorcet winning response?}
\end{center}

We now focus on mappings $\Psi$ that are continuous at 1/2, a condition commonly encountered in practical learning setups. Under this mild assumption, we answer this question in Theorem \ref{thm:diversity} and provide its proof in Section \ref{sec:proof-diversity}.

\begin{theorem}\label{thm:diversity}
Assume that the mapping $\Psi(\cdot)$ is continuous at $1/2$.
Assuming the Condorcet consistency of problem \eqref{eq:gen-nlhf}, then any Nash solution is mixed when there is no Condorcet winning response if and only if $\Psi(\cdot)$ satisfies 
    \begin{equation}\label{eq:thm-diversity}
        \Psi(t)+\Psi(1-t)\geqslant 2\Psi(1/2)\,,\forall\, 0\leqslant t\leqslant 1\text{ and }
        \Psi(t)<\Psi(1/2)\,, \forall \, 0\leqslant t< 1/2\,.
    \end{equation}
\end{theorem}

The first condition arises from the requirement of mixed strategies, while the second condition is a reduction of the condition inherited from Theorem \ref{thm:condorcet-consistent} under the assumption of Condorcet consistency and the first condition.

Choices of payoff functions are harder to characterize when we relax the continuity assumption. The following example investigate a special piece-wise constant mapping, which does not satisfy the first condition in Theorem \ref{thm:diversity}.

\begin{example}\label{example:continuous}
Let $M_{-}<\Psi(1/2)\leqslant M_{+}$ and take
\[
\Psi(t)=\begin{cases}
M_{-}\,, & 0\leqslant t< 1/2 \\
\Psi(1/2)\,, & t=1/2 \\
M_{+}\,,& 1/2< t\leqslant 1
\end{cases}\,.
\]
Then, any Nash solution is mixed when there is no Condorcet winning response.
\end{example}
\begin{proof}[Proof of Example \ref{example:continuous}]
We prove this conclusion by contradiction. Suppose that the Nash solution is $\bm{\delta}_{i^\star}$ for some $i^\star\in [n]$, and the Nash equilibrium is $(\bm{\delta}_{i^\star},\bm{\pi}^\star)$. As there is no Condorcet winning response, by definition, there exists $j^\prime$ such that $\mathcal{P}(y_{i^\star}\succ y_{j^\prime})<1/2$. Then we have 
\begin{equation}\label{eq:proof-example-continuous-1}
\mathcal{P}_\Psi(\bm{\delta}_{i^\star},\bm{\pi}^\star)=\min_{\bm{\pi}}\mathcal{P}_\Psi(\bm{\delta}_{i^\star},\bm{\pi})\leqslant \mathcal{P}_\Psi(\bm{\delta}_{i^\star},\bm{\delta}_{j^\prime})=\Psi\left(\mathcal{P}\left(y_{i^\star}\succ y_{j^\prime}\right)\right)=M_-\,.
\end{equation}
However, choosing $i^\prime$ such that $\pi^\star_{i^\prime}>0$, we have 
\[
\mathcal{P}_\Psi(\bm{\delta}_{i^\star},\bm{\pi}^\star)=\max_{\bm{\pi}}\mathcal{P}_\Psi(\bm{\pi},\bm{\pi}^\star)\geqslant\mathcal{P}_\Psi(\bm{\delta}_{i^\prime},\bm{\pi}^\star)=\sum_{i=1}^n\pi_i^\star\Psi\left(\mathcal{P}\left(y_{i^\prime}\succ y_i\right)\right)>M_-\,,
\]
which causes a contradiction to \eqref{eq:proof-example-continuous-1}. Hence, we conclude our proof.
\end{proof}

This example implies that choices of payoff functions are considerably richer when we relax the continuity assumption.

\subsection{Proof of Theorem \ref{thm:condorcet-consistent}}
\label{sec:proof-condorcet}

Without any loss of generality, we assume that $y_1$ is the Condorcet winning response. First, we show that a necessary condition that ensures the Condorcet consistency of problem \eqref{eq:gen-nlhf} is:
\begin{equation}\label{eq:necessary-condition-Condorcet}
    \begin{cases}
        \Psi(t)\geqslant \Psi(1/2)\,, 1\geqslant t\geqslant1/2\\\Psi(t)<\Psi(1/2)\,,0\leqslant t<1/2
    \end{cases}\,.
\end{equation}
To show this, we examine the case where $n=2$. For any $1\geqslant t> 1/2$, we consider the game with the payoff in Table \ref{table: condorcet consistency}. By the definition of Condorcet consistency, all Nash equilibrium of this game is of the form $(\bm{\delta}_1,\bm{\pi}^\star)$ for some $\bm{\pi}^\star$.

\begin{table}[htbp]
\caption{Payoff matrix with two responses $\{y_1,y_2\}$.}
\label{table: condorcet consistency}
\centering
\scalebox{1.0}{\begin{tabular}{c | c  c }
\toprule
          $\Psi(\cP(y\succ y'))$ &  
          $y'=y_1$ & 
          $y'=y_2$ 
         \\ \midrule
         $y=y_1$ &
         $\Psi(1/2)$ & $\Psi(t)$ \\
         $y=y_2$ &
         $\Psi(1-t)$ & $\Psi(1/2)$ \\ \bottomrule
    \end{tabular}}
\end{table}

\paragraph{Case 1.} If $\Psi(t)>\Psi(1/2)$, we have 
\[
\bm{\pi}^\star=\arg\min_{\bm{\pi}}\mathcal{P}_\Psi(\bm{\delta}_1,\bm{\pi})=\arg\min_{\bm{\pi}}\left\{\pi_1\Psi(1/2)+\pi_2\Psi(t)\right\}=\bm{\delta}_1\,.
\]
Therefore, we have
\[\begin{aligned}
&\Psi(1/2)=\mathcal{P}_\Psi(\boldsymbol{\delta}_1,\boldsymbol{\delta}_1)=\max_{\boldsymbol{\pi}}\mathcal{P}_\Psi(\boldsymbol{\pi},\boldsymbol{\delta}_1)\geqslant \mathcal{P}_\Psi(\boldsymbol{\delta}_2,\boldsymbol{\delta}_1)=\Psi_{21}=\Psi(1-t)\,,\\
&\Psi(1/2)=\mathcal{P}_\Psi(\boldsymbol{\delta}_1,\boldsymbol{\delta}_1)=\min_{\boldsymbol{\pi}}\mathcal{P}_\Psi(\boldsymbol{\delta}_1,\boldsymbol{\pi})\leqslant \mathcal{P}_\Psi(\boldsymbol{\delta}_1,\boldsymbol{\delta}_2)=\Psi_{12}=\Psi(t)\,.
\end{aligned}\]
Hence, we have $\Psi(1-t)\leqslant\Psi(1/2)< \Psi(t)$. If $\Psi(1/2)=\Psi(1-t)$, notice that 
\[\begin{aligned}
&\mathcal{P}_\Psi(\boldsymbol{\pi},\boldsymbol{\delta}_1)=\pi_1\Psi(1/2)+\pi_2\Psi(1-t)=\Psi(1/2)\Longrightarrow \boldsymbol{\delta}_2\in\arg\max_{\boldsymbol{\pi}}\mathcal{P}_\Psi(\boldsymbol{\pi},\boldsymbol{\delta}_1)\,,\\
&\mathcal{P}_\Psi(\boldsymbol{\delta}_2,\boldsymbol{\pi})=\pi_1\Psi(1-t)+\pi_2\Psi(1/2)=\Psi(1/2)\Longrightarrow \boldsymbol{\delta}_1\in\arg\min_{\boldsymbol{\pi}}\mathcal{P}_\Psi(\boldsymbol{\delta}_2,\boldsymbol{\pi})\,.
\end{aligned}\]
Therefore, $(\boldsymbol{\delta}_2,\boldsymbol{\delta}_1)$ is also a Nash equilibrium, which causes a contradiction to the fact that problem \eqref{eq:gen-nlhf} is Condorcet consistent. Therefore, we have $\Psi(t)>\Psi(1/2)>\Psi(1-t)$ for any $1\geqslant t>1/2$.

\paragraph{Case 2.} If $\Psi(t)<\Psi(1/2)$, we have 
\[
\bm{\pi}^\star=\arg\min_{\bm{\pi}}\mathcal{P}_\Psi(\bm{\delta}_1,\bm{\pi})=\arg\min_{\bm{\pi}}\left\{\pi_1\Psi(1/2)+\pi_2\Psi(t)\right\}=\bm{\delta}_2\,.
\]
However, notice that
\[
\Psi(1/2)=\mathcal{P}_\Psi(\bm{\delta}_2,\bm{\delta}_2)\leqslant\max_{\bm{\pi}}\mathcal{P}_\Psi(\bm{\pi},\bm{\delta}_2)=\mathcal{P}_\Psi(\bm{\delta}_1,\bm{\delta}_2)=\Psi(t)<\Psi(1/2)\,,
\]
which causes a contradiction.

\paragraph{Case 3.} If $\Psi(t)=\Psi(1/2)$. When $\Psi(1-t)=\Psi(1/2)$, any $(\bm{\pi}_1,\bm{\pi}_2)$ is a Nash equilibrium, which causes a contradiction to the fact that problem \eqref{eq:gen-nlhf} is Condorcet consistent. When $\Psi(1-t)>\Psi(1/2)$, note that
\[\begin{aligned}
\mathcal{P}_\Psi(\bm{\delta}_2,\bm{\pi})=\pi_1\Psi(1-t)+\pi_2\Psi(1/2)\geqslant\Psi(1/2)\Longrightarrow &\ \bm{\delta}_2\in \arg\min_{\bm{\pi}}\mathcal{P}_\Psi(\bm{\delta}_2,\bm{\pi})\,,\\
\mathcal{P}_\Psi(\bm{\pi},\bm{\delta}_2)=\pi_1\Psi(t)+\pi_2\Psi(1/2)=\Psi(1/2)\Longrightarrow &\ \bm{\delta}_2\in\arg\max_{\bm{\pi}}\mathcal{P}_\Psi(\bm{\pi},\bm{\delta}_2)\,. 
\end{aligned}\]
Therefore, $(\bm{\delta}_2,\bm{\delta}_2)$ is a Nash equilibrium, which also causes a contradiction to the fact problem \eqref{eq:gen-nlhf} is Condorcet consistent. Hence, we have $\Psi(1-t)<\Psi(1/2)$.

In summary, for any $1\geqslant t> 1/2$, we have $\Psi(1-t)<\Psi(1/2)\leqslant \Psi(t)$. 
Hence, \eqref{eq:necessary-condition-Condorcet} holds if problem \eqref{eq:gen-nlhf} is Condorcet consistent. 
Next, we prove that \eqref{eq:necessary-condition-Condorcet} is also sufficient for the Condorcet consistency of problem \eqref{eq:gen-nlhf}. Recall that $\Psi_{i 1}=\Psi(\mathcal{P}(y_i\succ y_1))<\Psi(1/2)$, and $\Psi_{1 i}=\Psi(\mathcal{P}(y_1\succ y_i))\geqslant\Psi(1/2)$ for any $i\neq 1$. If $(\boldsymbol{\pi}_1^\star,\boldsymbol{\pi}_2^\star)$ is a Nash equilibrium. Notice that
\[\begin{aligned}
&\mathcal{P}_\Psi(\boldsymbol{\pi}_1^\star,\boldsymbol{\pi}_2^\star)=\max_{\boldsymbol{\pi}}\mathcal{P}_\Psi(\boldsymbol{\pi},\boldsymbol{\pi}_2^\star)\geqslant \mathcal{P}_\Psi(\boldsymbol{\delta}_1,\boldsymbol{\pi}_2^\star)=\sum_{i=1}^n \pi^\star_{2,i}\Psi_{1i}\,,\\
&\mathcal{P}_\Psi(\boldsymbol{\pi}_1^\star,\boldsymbol{\pi}_2^\star)=\min_{\boldsymbol{\pi}}\mathcal{P}_\Psi(\boldsymbol{\pi}_1^\star,\boldsymbol{\pi})\leqslant \mathcal{P}_\Psi(\boldsymbol{\pi}_1^\star,\boldsymbol{\delta}_1)=\sum_{i=1}^n\pi_{1,i}^\star\Psi_{i1}\,.
\end{aligned}\]
Therefore, if $\boldsymbol{\pi}_1^\star\neq \boldsymbol{\delta}_1$, we have 
\begin{equation}\label{eq:uniqueness-proof-condorcet}
    \Psi(1/2)\leqslant \sum_{i=1}^n \pi^\star_{2,i}\Psi_{1i}\leqslant \mathcal{P}_\Psi(\boldsymbol{\pi}_1^\star,\boldsymbol{\pi}_2^\star)\leqslant \sum_{i=1}^n\pi_{1,i}^\star\Psi_{i1}<\Psi(1/2)\,,
\end{equation}
which causes a contradiction. Therefore, $\boldsymbol{\pi}_1^\star=\boldsymbol{\delta}_1$, i.e., problem \eqref{eq:gen-nlhf} is Condorcet consistent.

\subsection{Proof of Theorem \ref{thm:diversity}}
\label{sec:proof-diversity}

First, according to Theorem \ref{thm:condorcet-consistent}, when the Nash solution is Condorcet consistent, we have 
\begin{equation}\label{eq:diversity-condorcet-condition}
    \begin{cases}
        \Psi(t)\geqslant\Psi(1/2)\,, 1\geqslant t\geqslant 1/2\\ \Psi(t)<\Psi(1/2)\,, 1/2>t\geqslant 0
    \end{cases}\,.
\end{equation}

In addition, we show that $\Psi(\cdot)$ must satisfy $\Psi(t)+\Psi(1-t)\geqslant 2\Psi(1/2)\,,\forall t\in [0,1]$ for ensuring that the Nash solution is mixed when there is no Condorcet winning response. We consider the case where $n=4$ and the game with the payoff in Table \ref{table:diversity} for any $t_1,t_2>1/2$.

\begin{table}[htbp]
        \caption{Payoff matrix with four responses $\{y_1,y_2,y_3,y_4\}$.}
        \label{table:diversity}
        \centering
        \scalebox{1.0}{
        \begin{tabular}{c | c  c  c  c }
        \toprule
              $\Psi(\cP(y\succ y'))$ &  
              $y'=y_1$ & 
              $y'=y_2$ &
              $y'=y_3$ &
              $y'=y_4$ 
             \\ \midrule
             $y=y_1$ &
             $\Psi(1/2)$ & $\Psi(t_1)$ & $\Psi(1-t_1)$ & $\Psi(t_2)$ \\
             $y=y_2$ &
             $\Psi(1-t_1)$ & $\Psi(1/2)$ & $\Psi(t_1)$ & $\Psi(t_2)$\\
             $y=y_3$ &
             $\Psi(t_1)$ & $\Psi(1-t_1)$ & $\Psi(1/2)$ & $\Psi(t_2)$ \\
             $y=y_4$ &
             $\Psi(1-t_2)$ & $\Psi(1-t_2)$ & $\Psi(1-t_2)$ & $\Psi(1/2)$\\
             \bottomrule
        \end{tabular}}
    \end{table}

Notice that if $\Psi(t_1)+\Psi(1-t_1)+\Psi(1/2)\leqslant 3\Psi(1-t_2)$, we have 
\[
\mathcal{P}_\Psi(\bm{\delta}_4,\bm{\pi})= (\pi_1+\pi_2+\pi_3)\Psi(1-t_2) + \pi_4\Psi(1/2) \Longrightarrow \frac{\bm{\delta}_1+\bm{\delta}_2+\bm{\delta}_3}{3}\in\arg\min_{\bm{\pi}}\mathcal{P}_\Psi(\bm{\delta}_4,\bm{\pi})\,,
\]
and 
\[\begin{aligned}
\mathcal{P}_\Psi\left(\bm{\pi},\frac{\bm{\delta}_1+\bm{\delta}_2+\bm{\delta}_3}{3}\right)=&\ (\pi_1+\pi_2+\pi_3)\cdot\frac{\Psi(1/2)+\Psi(t_1)+\Psi(1-t_1)}{3}+\pi_4\Psi(1-t_2) &\ \\
&\ \Longrightarrow \bm{\delta}_4\in\arg\max_{\bm{\pi}}\mathcal{P}_\Psi\left(\bm{\pi},\frac{\bm{\delta}_1+\bm{\delta}_2+\bm{\delta}_3}{3}\right)\,.
\end{aligned}\]

Therefore, $(\bm{\delta}_4,(\bm{\delta}_1+\bm{\delta}_2+\bm{\delta}_3)/3)$ is a Nash equilibrium, which causes a contradiction to the fact that the Nash solution is mixed. Hence, we have $\Psi(t_1)+\Psi(1-t_1)+\Psi(1/2)>3\Psi(1-t_2)$ for any $t_1,t_2>1/2$. Let $t_2\to 1/2$, we have $\Psi(t)+\Psi(1-t)\geqslant 2\Psi(1/2)$ for any $t\in [0,1]$. Hence, combining \eqref{eq:diversity-condorcet-condition}, we have shown that the necessary condition for ensuring that the Nash solution is mixed is:
\begin{equation}\label{eq:diversity-necessary-condition}
    \Psi(t)+\Psi(1-t)\geqslant 2\Psi(1/2)\,,\forall t\in [0,1]\text{ and }\begin{cases}
        \Psi(t)\geqslant\Psi(1/2)\,, 1\geqslant t\geqslant 1/2\\ \Psi(t)<\Psi(1/2)\,, 1/2>t\geqslant 0
    \end{cases}\,.
\end{equation}
Next, we prove that the condition \eqref{eq:diversity-necessary-condition} is also sufficient. Suppose that $(\bm{\delta}_{i^\star},\bm{\pi}^\star)$ is a Nash equilibrium, then we have 
\[\begin{aligned}
\mathcal{P}_\Psi(\bm{\delta}_{i^\star},\bm{\pi}^\star)=\max_{\bm{\pi}}\mathcal{P}_\Psi(\bm{\pi},\bm{\pi}^\star)\geqslant &\ \mathcal{P}_\Psi(\bm{\pi}^\star,\bm{\pi}^\star)\\
    =&\ \sum_{i=1}^n\sum_{j=1}^n\pi^\star_i\pi^\star_j\Psi_{ij}=\frac{1}{2}\sum_{i=1}^n\sum_{j=1}^n\pi_i^\star\pi_j^\star\left(\Psi_{ij}+\Psi_{ji}\right)\geqslant \Psi(1/2)\,.
\end{aligned}\]
However, notice that for any $j$, we have 
\[
\mathcal{P}_\Psi(\bm{\delta}_{i^\star},\bm{\pi}^\star)=\min_{\bm{\pi}}\mathcal{P}_\Psi(\bm{\delta}_{i^\star},\bm{\pi})\leqslant \mathcal{P}_\Psi(\bm{\delta}_{i^\star},\bm{\delta}_j)=\Psi_{i^\star j}\,.
\]
As there is no Condorcet winning response, there must exist $j^\star$ such that $\mathcal{P}(y_{i^\star}\succ y_{j^\star})<1/2$, thus $\Psi_{i^\star j^\star}<\Psi(1/2)$. Hence, $\Psi(1/2)\leqslant \mathcal{P}_\Psi(\bm{\delta}_{i^\star},\bm{\pi}^\star)\leqslant \Psi_{i^\star j^\star}<\Psi(1/2)$, which causes a contradiction. Therefore, the Nash solution must be mixed.

\section{Smith Consistency}\label{sec:smith}

In this section, we extend the discussion of Condorcet consistency to Smith consistency. First, we define the Smith set and Smith consistency. Next, we present Theorem \ref{thm:smith}, which provides the necessary and sufficient condition for the mapping $\Psi$ to ensure Smith consistency. Finally, we highlight that Smith-consistent methods inherently preserve the diversity present in human preferences and discuss the continuity assumption in Theorem \ref{thm:smith}. We also discuss the continuity assumption used in Theorem \ref{thm:smith} through Example \ref{example:continuous-Smith}. The proofs of Theorem \ref{thm:smith} and Example \ref{example:continuous-Smith} are provided in Section \ref{sec:proof-smith} and Section \ref{sec:proof-continuous-Smith}, respectively.

Condorcet consistency only ensures the method capture the right response when there exists a Condorcet winning response. In general, when there is no Condorcet winning response, we can expect that there might be a set of responses satisfying similar property, generalizing Definition \ref{def:Condorcet winning response}.
Under Assumption \ref{ass:strict_preference}, 
\citet{liu2025statistical} revealed a more detailed decomposition of the preference structure. 
Specifically, the set of responses can be partitioned into distinct groups $S_1, \ldots, S_k$, where every response in $S_i$ is preferred over all responses in $S_j$ for $i < j$, summarized in the following theorem. 

\begin{theorem}[\cite{liu2025statistical}]\label{thm:preference-decomposition}
    Under Assumption~\ref{ass:strict_preference}, the set of responses can be partitioned into disjoint subsets $S_1,\ldots,S_k$ such that: 
\begin{enumerate}
\item Each $S_i$ either forms a Condorcet cycle or is a single response.
\item For any $j>i$, any response $y\in S_i$ and $y'\in S_j$, $\cP(y\succ y')>\frac{1}{2}$.
\end{enumerate}
Moreover, this decomposition is unique.
\end{theorem}

When $|S_1|=1$, the response in $S_1$ is exactly the Condorcet winning response.
Thus, $S_1$ is the generalization of Condorcet winning response, and is referred as the Condorcet winning set in \citet{liu2025statistical}. 
Traditionally, a subset with such property is also known as the Smith set in the literature of social choice theory \citep{shoham2008multiagent}.
Here we choose to adopt the name Smith set to distinguish with the concept of Condorcet winning response.
Given this decomposition, it is natural to desire that an aligned LLM adopts a strategy supported exclusively on the top group $S_1$, as any response outside $S_1$ is strictly less preferred than any response inside $S_1$. 
This desirable property is referred to as Smith consistency:

\begin{definition}[Smith Consistency]\label{def: Smith consistency}
    Problem \eqref{eq:gen-nlhf} is Smith consistent if the support of any Nash solution is contained in the Smith set $S_1$.
\end{definition}

\citet{liu2025statistical} showed that the original \texttt{NLHF} payoff, which corresponds to the case where $\Psi(t)=t$, is Smith consistent. Here, we investigate this question for a general mapping $\Psi$:
\begin{center}
\emph{ Which choices of $\Psi$ ensure Smith consistency?}
\end{center}

Here, similar to Theorem \ref{thm:diversity}, we answer this question in Theorem \ref{thm:smith} for mappings that is continuous at $1/2$. The proof is provided in Section \ref{sec:proof-smith}. 

\begin{theorem}\label{thm:smith}
Suppose that the mapping $\Psi(\cdot)$ is continuous at $1/2$, problem \eqref{eq:gen-nlhf} is Smith consistent if and only if $\Psi(\cdot)$ satisfies
\[
\Psi(t)+\Psi(1-t)=2\Psi(1/2)\,,\forall \, t\in [0,1]\text{ and }\Psi(t)<\Psi(1/2)\,,\forall\, 0\leqslant t< 1/2\,.
\]
\end{theorem}
The first condition $\Psi(t)+\Psi(1-t)=2\Psi(1/2)$ says nothing but the zero-sum game formed by problem \eqref{eq:gen-nlhf} is equivalent to a symmetric two-player zero-sum game\footnote{This can be seen by shifting the payoff by $\Psi(1/2)$, which leaves the Nash solution unchanged.} \citep{duersch2012pure}.
By definition, Smith consistency implies Condorcet consistency because when there is a Condorcet winning response, $S_1$ is exactly the set whose only element is the Condorcet winning response. Thus, the second condition is just a reduction of the condition in Theorem \ref{thm:condorcet-consistent} under the first condition.
It is easy to see $\Psi(t)=t$ satisfies these conditions, and thus our result generalize Theorem 3.6 in \citet{liu2025statistical}.
More interestingly, $\Psi(t)=\log(t/(1-t))$ also satisfies these conditions.
This implies that 
\begin{equation*}
\max_{\bm{\pi}} \min_{\bm{\pi}^{\prime}} \mathbb{E}_{x\sim\rho} \left[\mathbb{E}_{y\sim \bm{\pi}(\cdot\mid x)}\mathbb{E}_{y^\prime\sim \bm{\pi}^\prime(\cdot\mid x)}\left[\log\left(\frac{\mathcal{P}(y\succ y^{\prime}\mid x)}{\mathcal{P}(y^{\prime}\succ y\mid x)}\right)\right]\right],
\end{equation*}
which is a natural generalization of standard \texttt{RLHF} when human preferences does not satisfy BTL model, is also Smith consistent.

The set of choices for $\Psi$ that ensure Smith consistency is quite broad. We can easily construct such a $\Psi$ by first defining $\Psi(t)$ on $[0, 1/2]$ to satisfy $\Psi(t) < \Psi(1/2)$ for all $t \in [0, 1/2)$, and then extending it to $[0, 1]$ by setting $\Psi(t) = 2\Psi(1/2) - \Psi(1-t)$ for all $t \in (1/2, 1]$.
Moreover, as discussed in Section \ref{sec:condorcet}, a practical preference model $\mathcal{P}_\theta(y \succ y^\prime)$ can be seen as a mapping of the ground truth preference via $\Psi$, i.e., $\Psi(\mathcal{P}(y \succ y^\prime))$. Thus, the first condition in Theorem \ref{thm:smith} requires the preference model to satisfy $\mathcal{P}_\theta(y \succ y^\prime) + \mathcal{P}_\theta(y^\prime \succ y) = 1$, with $\mathcal{P}_\theta(y \succ y) = 1/2$ enforced.
However, several practically used preference models \citep{munos2024nash, jiang-etal-2023-llm, wu2024self} do not guarantee this condition, which may cause the aligned LLM strategy to fail to satisfy Smith consistency.

As any mapping satisfying the condition in Theorem \ref{thm:smith} also satisfies the condition in Theorem \ref{thm:diversity}, we obtain the following corollary:
\begin{corollary}\label{cor:smith diverse}
    Suppose that the mapping $\Psi(\cdot)$ is continuous at $1/2$. Then if problem \eqref{eq:gen-nlhf} is Smith consistent, any Nash solution is also mixed.
\end{corollary}
This shows that when $|S_1|>1$, the Nash solution to problem \eqref{eq:gen-nlhf} with any $\Psi$ such that Smith consistency holds will not only support on $S_1$ but also be a mixed strategy on $S_1$ without collapsing to a single response. As a conclusion, a Smith consistent method can preserve the diversity inherent in human preferences, at least partially.

Lastly, we discuss what happens if $\Psi$ is not continuous at $1/2$. Choices of mappings $\Psi$ are considerably richer and consequently harder to characterize when we relax the continuity assumption.
The following example shows that the piece-wise constant mapping in Example \ref{example:continuous} also ensures Smith consistency.

\begin{example}\label{example:continuous-Smith}
Let $M_{-}<\Psi(\frac{1}{2})< M_{+}$, and we take 
\[
\Psi(t)=\begin{cases}
M_{-} & 0\leqslant t< 1/2 \\
\Psi(1/2)   & t=1/2 \\
M_{+} & 1/2< t\leqslant 1 
\end{cases}\,.
\]
Then problem \eqref{eq:gen-nlhf} is Smith consistent. The proof is provided in Section \ref{sec:proof-continuous-Smith}.
\end{example}

\subsection{Proof of Theorem \ref{thm:smith}}
\label{sec:proof-smith}

First, we show that the necessary condition for ensuring that problem \eqref{eq:gen-nlhf} is Smith consistent is:
\begin{equation}\label{eq:necessary-smith}
    \Psi(t)+\Psi(1-t)=2\Psi(1/2)\,,\forall t\in [0,1]\text{ and }\begin{cases}
        \Psi(t)\geqslant \Psi(1/2)\,,1\geqslant t\geqslant 1/2\\
        \Psi(t)<\Psi(1/2)\,, 0\leqslant t< 1/2
    \end{cases}\,.
\end{equation}
First, Condorcet consistency must hold when Smith consistency holds. According to Theorem \ref{thm:condorcet-consistent}, we have 
\[\begin{cases}
    \Psi(t)\geqslant \Psi(1/2)\,,1\geqslant t\geqslant 1/2\\
    \Psi(t)<\Psi(1/2)\,, 0\leqslant t< 1/2
\end{cases}\]
Next, we show that when $\Psi(\cdot)$ is continuous at $1/2$, $\Psi(\cdot)$ must satisfy $\Psi(t)+\Psi(1-t)\geqslant 2\Psi(1/2)$ (Lemma \ref{lem:smith-geqslant}) and $\Psi(t)+\Psi(1-t)\leqslant 2\Psi(1/2)$ (Lemma \ref{lem:smith-leqslant}) for any $t\in [0,1]$. Therefore, combining the two results together, we obtain the condition \eqref{eq:necessary-smith}.

\begin{lemma}\label{lem:smith-geqslant}
    When $\Psi(\cdot)$ is continuous at $1/2$. Achieving Smith consistency only if 
    \[\Psi(t)+\Psi(1-t)\geqslant 2\Psi(1/2)\,,\forall t\in [0,1]\,.\]
\end{lemma}

\begin{proof}[Proof of Lemma \ref{lem:smith-geqslant}]
    We consider the case where $n=4$ and the game with the payoff in Table \ref{table:necessary-smith-geqslant} for any $t_1,t_2>1/2$. Notice that if $\Psi(t_1)+\Psi(1-t_1)+\Psi(1/2)\leqslant 3\Psi(1-t_2)$, we have 
    \[
    \mathcal{P}_\Psi(\bm{\delta}_4,\bm{\pi})= (\pi_1+\pi_2+\pi_3)\Psi(1-t_2) + \pi_4\Psi(1/2) \Longrightarrow \frac{\bm{\delta}_1+\bm{\delta}_2+\bm{\delta}_3}{3}\in\arg\min_{\bm{\pi}}\mathcal{P}_\Psi(\bm{\delta}_4,\bm{\pi})\,,
    \]
    and 
    \[\begin{aligned}
    \mathcal{P}_\Psi\left(\bm{\pi},\frac{\bm{\delta}_1+\bm{\delta}_2+\bm{\delta}_3}{3}\right)=&\ (\pi_1+\pi_2+\pi_3)\cdot\frac{\Psi(1/2)+\Psi(t_1)+\Psi(1-t_1)}{3}+\pi_4\Psi(1-t_2) &\ \\
    &\ \Longrightarrow \bm{\delta}_4\in\arg\max_{\bm{\pi}}\mathcal{P}_\Psi\left(\bm{\pi},\frac{\bm{\delta}_1+\bm{\delta}_2+\bm{\delta}_3}{3}\right)\,.
    \end{aligned}\]
    Therefore, $(\bm{\delta}_4,(\bm{\delta}_1+\bm{\delta}_2+\bm{\delta}_3)/3)$ is a Nash equilibrium, which causes a contradiction to the fact that the Nash solution supports on $S_1:=\{y_1,y_2,y_3\}$. Hence, we have $\Psi(t_1)+\Psi(1-t_1)+\Psi(1/2)>3\Psi(1-t_2)$ for any $t_1,t_2>1/2$. Let $t_2\to 1/2$, we have $\Psi(t)+\Psi(1-t)\geqslant 2\Psi(1/2)$ for any $t\in [0,1]$.    
    \begin{table}[htbp]
        \caption{Payoff matrix with four responses $\{y_1,y_2,y_3,y_4\}$.}
        \label{table:necessary-smith-geqslant}
        \centering
        \scalebox{1.0}{
        \begin{tabular}{c | c  c  c  c }
        \toprule
              $\Psi(\cP(y\succ y'))$ &  
              $y'=y_1$ & 
              $y'=y_2$ &
              $y'=y_3$ &
              $y'=y_4$ 
             \\ \midrule
             $y=y_1$ &
             $\Psi(1/2)$ & $\Psi(t_1)$ & $\Psi(1-t_1)$ & $\Psi(t_2)$ \\
             $y=y_2$ &
             $\Psi(1-t_1)$ & $\Psi(1/2)$ & $\Psi(t_1)$ & $\Psi(t_2)$\\
             $y=y_3$ &
             $\Psi(t_1)$ & $\Psi(1-t_1)$ & $\Psi(1/2)$ & $\Psi(t_2)$ \\
             $y=y_4$ &
             $\Psi(1-t_2)$ & $\Psi(1-t_2)$ & $\Psi(1-t_2)$ & $\Psi(1/2)$\\
             \bottomrule
        \end{tabular}}
    \end{table}
\end{proof}

\begin{lemma}\label{lem:smith-leqslant}
    When $\Psi(\cdot)$ is continuous at $1/2$. Achieving Smith consistency only if 
    \[\Psi(t)+\Psi(1-t)\leqslant 2\Psi(1/2)\,,\forall t\in [0,1]\,.\]
\end{lemma}

\begin{proof}[Proof of Lemma \ref{lem:smith-leqslant}]
    We consider the case where $n=6$ and the game with the payoff in Table \ref{table:necessary-smith-leqslant} for any $t_1,t_2>1/2$. 
    Notice that if $\Psi(t_1)+\Psi(1/2)+\Psi(1-t_1)>3\Psi(t_2)(\geqslant3\Psi(1/2)>3\Psi(1-t_2))$, there exists positive $\bm{\mu}=(\mu_1/3,\mu_1/3,\mu_1/3,\mu_2/3,\mu_2/3,\mu_2/3)$ and $\bm{\mu}^\prime=(\mu_1^\prime/3,\mu_1^\prime/3,\mu_1^\prime/3\mu_2^\prime/3,\mu_2^\prime/3,\mu_2^\prime/3)$ such that $\mu_1+\mu_2=\mu^\prime_1+\mu^\prime_2=1$, and
    \[\begin{aligned}
    &\mu_1\left[\Psi(1/2)+\Psi(t_1)+\Psi(1-t_1)-3\Psi(t_2)\right] =\mu_2\left[\Psi(1/2)+\Psi(t_1)+\Psi(1-t_1)-3\Psi(1-t_2)\right]\,,\\
    &\mu^\prime_1\left[\Psi(1/2)+\Psi(t_1)+\Psi(1-t_1)-3\Psi(1-t_2)\right] =\mu^\prime_2\left[\Psi(1/2)+\Psi(t_1)+\Psi(1-t_1)-3\Psi(t_2)\right]\,.
    \end{aligned}\]
    Hence, we have 
    \[\begin{aligned}
    &\ \mu_1(\Psi(1/2)+\Psi(t_1)+\Psi(1-t_1))+3\mu_2\Psi(1-t_2)\\
    =&\ \mu_2(\Psi(1/2)+\Psi(t_1)+\Psi(1-t_1))+3\mu_1\Psi(t_2):=3A\,,\\
    &\ \mu^\prime_1(\Psi(1/2)+\Psi(t_1)+\Psi(1-t_1))+3\mu^\prime_2\Psi(t_2)\\
    =&\ \mu^\prime_2(\Psi(1/2)+\Psi(t_1)+\Psi(1-t_1))+3\mu^\prime_1\Psi(1-t_2):=3B\,.
    \end{aligned}\]
    Thus, we have 
    \[\begin{aligned}
        \mathcal{P}_\Psi(\bm{\pi},\bm{\mu}^\prime)=&\ \left(\pi_1+\pi_2+\pi_3\right)\left[\frac{\mu_1^\prime}{3}\left(\Psi(1/2)+\Psi(t_1)+\Psi(1-t_1)\right)+\mu_2^\prime\Psi(t_2)\right]\\
        &\ \, + \left(\pi_4+\pi_5+\pi_6\right)\left[\mu_1^\prime\Psi(1-t_2)+\frac{\mu_2^\prime}{3}\left(\Psi(1/2)+\Psi(t_1)+\Psi(1-t_1)\right)\right]=B\,,\\
    \end{aligned}\]
    and 
    \[\begin{aligned}
        \mathcal{P}_\Psi(\bm{\mu},\bm{\pi})=&\ \left(\pi_1+\pi_2+\pi_3\right)\left[\frac{\mu_1}{3}\left(\Psi(1/2)+\Psi(t_1)+\Psi(1-t_1)\right)+\mu_2\Psi(1-t_2)\right]\\
        &\ \,  + \left(\pi_4+\pi_5+\pi_6\right)\left[\mu_1\Psi(t_2)+\frac{\mu_2}{3}\left(\Psi(1/2)+\Psi(t_1)+\Psi(1-t_1)\right)\right]=A\,.
    \end{aligned}\]
    Therefore, $\bm{\mu}\in\arg\max_{\bm{\pi}}\mathcal{P}_\Psi(\bm{\pi},\bm{\mu}^\prime)\,,\,\bm{\mu}^\prime\in\arg\min_{\bm{\pi}}\mathcal{P}_\Psi(\bm{\mu},\bm{\pi})$, which provides that $(\bm{\mu},\bm{\mu}^\prime)$ is a Nash equilibrium. However, this causes a contradiction to the fact that the Nash solution supports on $S_1:=\{y_1,y_2,y_3\}$. Thus, it must hold that $\Psi(t_1)+\Psi(1/2)+\Psi(1-t_1)\leqslant 3\Psi(t_2)$ for any $t_1,t_2>1/2$. Let $t_2\to 1/2$, we obtain $\Psi(t)+\Psi(1-t)\leqslant 2\Psi(1/2)$ for any $t\in [0,1]$.
    \begin{table}[htbp]
    \caption{Payoff matrix with six responses $\{y_1,y_2,y_3,y_4,y_5,y_6\}$.}
    \label{table:necessary-smith-leqslant}
    \centering
    \scalebox{1.0}{
    \begin{tabular}{c | c  c  c  c  c  c }
        \toprule
          $\Psi(\cP(y\succ y'))$ &  
          $y'=y_1$ & 
          $y'=y_2$ &
          $y'=y_3$ &
          $y'=y_4$ &
          $y'=y_5$ &
          $y'=y_6$
         \\ \midrule
         $y=y_1$ &
         $\Psi(1/2)$ & $\Psi(t_1)$ & $\Psi(1-t_1)$ & $\Psi(t_2)$ & $\Psi(t_2)$ & $\Psi(t_2)$ \\
         $y=y_2$ &
         $\Psi(1-t_1)$ & $\Psi(1/2)$ & $\Psi(t_1)$ & $\Psi(t_2)$ & $\Psi(t_2)$ & $\Psi(t_2)$ \\
         $y=y_3$ &
         $\Psi(t_1)$ & $\Psi(1-t_1)$ & $\Psi(1/2)$ & $\Psi(t_2)$ & $\Psi(t_2)$ & $\Psi(t_2)$ \\
         $y=y_4$ &
         $\Psi(1-t_2)$ & $\Psi(1-t_2)$ & $\Psi(1-t_2)$ & $\Psi(1/2)$ & $\Psi(t_1)$ & $\Psi(1-t_1)$\\
         $y=y_5$ &
         $\Psi(1-t_2)$ & $\Psi(1-t_2)$ & $\Psi(1-t_2)$ &$\Psi(1-t_1)$ & $\Psi(1/2)$ & $\Psi(t_1)$  \\
         $y=y_6$ & 
         $\Psi(1-t_2)$ & $\Psi(1-t_2)$ & $\Psi(1-t_2)$ & $\Psi(t_1)$ & $\Psi(1-t_1)$ & $\Psi(1/2)$ \\
         \bottomrule
    \end{tabular}
    }
\end{table}
\end{proof}

Finally, we prove that the condition \eqref{eq:necessary-smith} is also sufficient for Smith consistency. Suppose that $(\bm{\pi}_1^\star,\bm{\pi}_2^\star)$ is a Nash equilibrium, notice that
\begin{equation}\label{eq:factI-smith-sufficient}
    \begin{aligned}
        &\mathcal{P}_\Psi(\bm{\pi}_1^\star,\bm{\pi}_2^\star)=\max_{\bm{\pi}}\mathcal{P}_\Psi(\bm{\pi},\bm{\pi}_2^\star)\geqslant \mathcal{P}_\Psi(\bm{\pi}_2^\star,\bm{\pi}_2^\star)=\Psi(1/2)\,,\\
        &\mathcal{P}_\Psi(\bm{\pi}_1^\star,\bm{\pi}_2^\star)=\min_{\bm{\pi}}\mathcal{P}_\Psi(\bm{\pi}_1^\star,\bm{\pi})\leqslant \mathcal{P}_\Psi(\bm{\pi}_1^\star,\bm{\pi}_1^\star)=\Psi(1/2)\,,
    \end{aligned}
\end{equation}
which follows from the following fact: for any $\bm{\pi}$,
\[
\mathcal{P}_\Psi(\bm{\pi},\bm{\pi})=\sum_{i=1}^n\sum_{j=1}^n\pi_i\pi_j\Psi_{ij}=\frac{1}{2}\sum_{i=1}^n\sum_{j=1}^n\pi_i\pi_j\left(\Psi_{ij}+\Psi_{ji}\right)=\Psi(1/2)\sum_{i=1}^n\sum_{j=1}^n\pi_i\pi_j=\Psi(1/2)\,.
\]
Thus, from \eqref{eq:factI-smith-sufficient}, we have $\mathcal{P}_\Psi(\bm{\pi}_1^\star,\bm{\pi}_2^\star)=\Psi(1/2)$. Then we prove $\operatorname{supp}(\bm{\pi}_1^\star)\subseteq S_1$. Hence, the Nash solution is Smith consistent, i.e., only supports on $S_1$.
\paragraph{Case 1.} If $\operatorname{supp}(\bm{\pi}_1^\star)\bigcap S_1=\emptyset$, taking any $j\in S_1$, we have 
\[
\mathcal{P}_\Psi(\bm{\pi}_1^\star,\bm{\pi}_2^\star)=\min_{\bm{\pi}}\mathcal{P}_\Psi(\bm{\pi}_1^\star,\bm{\pi})\leqslant \mathcal{P}_\Psi(\bm{\pi}_1^\star,\bm{\delta}_j)=\sum_{i=1}^n\pi_{1,i}^\star\Psi_{ij}=\sum_{i\in S_1^c}\pi_{1,i}^\star\Psi_{ij}<\Psi(1/2)\,,
\]which causes a contradiction to the fact that $\mathcal{P}_\Psi(\bm{\pi}_1^\star,\bm{\pi}_2^\star)=\Psi(1/2)$.
\paragraph{Case 2.} If $\operatorname{supp}(\bm{\pi}_1^\star)\bigcap S_1\neq\emptyset$, and $\operatorname{supp}(\bm{\pi}_1^\star)\bigcap S_1^c\neq \emptyset$, taking $\widetilde{\bm{\pi}}_2^\star$ as: 
\[\widetilde{\pi}^\star_{2,j}=\mathds{1}\left\{j\in S_1\right\}\cdot \frac{\pi^\star_{1,j}}{\sum_{j\in S_1}\pi^\star_{1,j}}\,.\]
Then we have 
\begin{equation}\label{eq:factII-smith-sufficient}
    \begin{aligned}
        \mathcal{P}_\Psi(\bm{\pi}_1^\star,\bm{\pi}_2^\star)=&\ \min_{\bm{\pi}}\mathcal{P}_\Psi(\bm{\pi}_1^\star,\bm{\pi})\leqslant \mathcal{P}_\Psi(\bm{\pi}_1^\star,\widetilde{\bm{\pi}}^\star_2)\\
        = &\ \sum_{i\in S_1}\sum_{j\in S_1}\pi^\star_{1,i}\widetilde{\pi}^\star_{2,j}\Psi_{ij}+\sum_{i\in S_1^c}\sum_{j\in S_1}\pi^\star_{1,i}\widetilde{\pi}^\star_{2,j}\Psi_{ij}\\
        < &\ \frac{\sum_{i\in S_1}\sum_{j\in S_1}\pi^\star_{1,i}\pi^\star_{1,j}\Psi_{ij}}{\sum_{j\in S_1}\pi^\star_{1,j}} + \Psi(1/2)\sum_{i\in S_1^c}\sum_{j\in S_1}\pi^\star_{1,i}\widetilde{\pi}^\star_{2,j}\\
        = &\ \Psi(1/2)\sum_{i\in S_1}\pi_{1,i}^\star+\Psi(1/2)\sum_{i\in S_1^c}\pi^\star_{1,i}=\Psi(1/2)\,,
    \end{aligned}
\end{equation}
which follows from the following fact:
\[\begin{aligned}
    \sum_{i\in S_1}\sum_{j\in S_1}\pi_{1,i}^\star\pi_{1,j}^\star\Psi_{ij}=&\ \frac{1}{2}\sum_{i\in S_1}\sum_{j\in S_1}\pi_{1,i}^\star\pi_{1,j}^\star\left(\Psi_{ij}+\Psi_{ji}\right)\\
    =&\ \Psi(1/2)\sum_{i\in S_1}\sum_{j\in S_1}\pi_{1,i}^\star\pi_{1,j}^\star=\Psi(1/2)\left(\sum_{i\in S_1}\pi_{1,i}^\star\right)\left(\sum_{j\in S_1}\pi_{1,j}^\star\right)
\end{aligned}\]
However, \eqref{eq:factII-smith-sufficient} also causes a contradiction to the fact that $\mathcal{P}_\Psi(\bm{\pi}_1^\star,\bm{\pi}_2^\star)=\Psi(1/2)$.

Therefore, it must hold that $\operatorname{supp}(\bm{\pi}_1^\star)\bigcap S_1^c= \emptyset$, i.e., $\operatorname{supp}(\bm{\pi}_1^\star)\subseteq S_1$.

\subsection{Proof of Example \ref{example:continuous-Smith}}\label{sec:proof-continuous-Smith}

We prove this conclusion by contradiction. Suppose that the Nash solution is $\bm{\pi}_1^\star$ that satisfies $\operatorname{supp}(\bm{\pi}_1^\star)\cap S_1^c\neq\emptyset$, and the Nash equilibrium is $(\bm{\pi}_1^\star,\bm{\pi}_2^\star)$. 

\paragraph{Case 1.} If $\operatorname{supp}(\bm{\pi}_1^\star)\cap S_1=\emptyset$, taking $j^\prime\in S_1$, we have 
\[
\mathcal{P}_\Psi(\bm{\pi}_1^\star,\bm{\pi}_2^\star)=\min_{\bm{\pi}}\mathcal{P}_\Psi(\bm{\pi}_1^\star,\bm{\pi})\leqslant \mathcal{P}_\Psi(\bm{\pi}_1^\star,\bm{\delta}_{j^\prime})=\sum_{i\in S_1^c}\pi^\star_{1,i}\Psi_{ij^\prime}=M_-\,.
\]
However, we have 
\[
\mathcal{P}_\Psi(\bm{\pi}_1^\star,\bm{\pi}_2^\star)=\max_{\bm{\pi}}\mathcal{P}_\Psi(\bm{\pi},\bm{\pi}_2^\star)\geqslant \mathcal{P}_\Psi(\operatorname{Unif}(S_1),\bm{\pi}_2^\star)=\sum_{i\in S_1}\sum_{j=1}^n \frac{\pi^\star_{2,j}}{\vert S_1\vert}\Psi_{ij}>M_-\,,
\]
which causes a contradiction. 

\paragraph{Case 2.} If $\operatorname{supp}(\bm{\pi}_2^\star)\cap S_1=\emptyset$ and $\operatorname{supp}(\bm{\pi}_1^\star)\cap S_1\neq\emptyset$, taking $i^\prime\in\operatorname{supp}(\bm{\pi}_1^\star)\cap S_1$, we have 
\[
\mathcal{P}_\Psi(\bm{\pi}_1^\star,\bm{\pi}_2^\star)=\max_{\bm{\pi}}\mathcal{P}_\Psi(\bm{\pi},\bm{\pi}_2^\star)\geqslant \mathcal{P}_\Psi(\bm{\delta}_{i^\prime},\bm{\pi}_2^\star)=\sum_{j\in S_1^c}\pi_{2,j}^\star\Psi_{i^\prime j}=M_+\,.
\]
However, we have 
\[
\mathcal{P}_\Psi(\bm{\pi}_1^\star,\bm{\pi}_2^\star)=\min_{\bm{\pi}}\mathcal{P}_\Psi(\bm{\pi}_1^\star,\bm{\pi})\leqslant \mathcal{P}_\Psi(\bm{\pi}_1^\star,\operatorname{Unif}(S_1))=\sum_{i=1}^n\sum_{j\in S_1}\frac{\pi^\star_{1,i}}{\vert S_1\vert}\Psi_{ij}<M_+\,,
\]
which cause a contradiction.

\paragraph{Case 3.} If If $\operatorname{supp}(\bm{\pi}_2^\star)\cap S_1\neq\emptyset$ and $\operatorname{supp}(\bm{\pi}_1^\star)\cap S_1\neq\emptyset$, taking $i_2^\star\in \operatorname{supp}(\bm{\pi}_2^\star)\cap S_1$, we consider the following strategy $\bm{\pi}^\prime_1$:
\[
\begin{cases}
    \pi^\prime_{1,i}=0\,, & i\in S_1^c\\
    \pi^\prime_{1,i}=\pi^\star_{1,i}\,, & i\in S_1\backslash \{i_2^\star\}\\
    \pi^\prime_{1,i_2^\star}=\pi^\star_{1,i_2^\star}+\sum_{i\in S_1^c}\pi^\star_{1,i}\,, & i=i_2^\star
\end{cases}\,.
\]
Then we have 
\begin{equation}\label{eq:important-ineq-example-smith-continuous}
    \begin{aligned}
    \mathcal{P}_\Psi(\bm{\pi}^\prime_1,\bm{\pi}_2^\star)-\mathcal{P}_\Psi(\bm{\pi}_1^\star,\bm{\pi}_2^\star)=&\ \sum_{i=1}^n\sum_{j=1}^n \left(\pi^\prime_{1,i}-\pi^\star_{1,i}\right)\pi^\star_{2,j}\Psi_{ij}\\
    = &\ -\sum_{i\in S_1^c}\sum_{j=1}^n \pi^\star_{1,i}\pi^\star_{2,j}\Psi_{ij} + \sum_{j=1}^n\sum_{i\in S_1^c}\pi_{1,i}^\star\pi^\star_{2,j}\Psi_{i_2^\star j}\\
    = &\ \sum_{j=1}^n \pi^\star_{2,j}\left[\sum_{i\in S_1^c}\pi_{1,i}^\star\left(\Psi_{i_2^\star j}-\Psi_{ij}\right)\right]>0\,.
\end{aligned}
\end{equation}
where the last inequality follows from the following two facts: for any $i\in S_1^c$, 
\[\Psi_{i_2^\star j}-\Psi_{ij}=\begin{cases}
    M_+-\Psi_{ij}\geqslant 0\,, & j\in S_1^c\\
    \Psi_{i_2^\star j}-M_-\geqslant 0\,, & j\in S_1
\end{cases}\,,\]
and when $j=i_2^\star$, 
\[
\pi^\star_{2,i_2^\star}\left[\sum_{i\in S_1^c}\pi_{1,i}^\star\left(\Psi_{i_2^\star i_2^\star}-\Psi_{i i_2^\star}\right)\right]=\pi^\star_{2,i_2^\star}\left(\Psi(1/2)-M_-\right)\sum_{i\in S_1^c}\pi^\star_{1,i}>0\,.
\]
However, \eqref{eq:important-ineq-example-smith-continuous} causes a contradiction to the fact that $\mathcal{P}_\Psi(\bm{\pi}_1^\prime,\bm{\pi}_2^\star)\leqslant \max_{\bm{\pi}}\mathcal{P}_\Psi(\bm{\pi},\bm{\pi}^\star_2)=\mathcal{P}_\Psi(\bm{\pi}_1^\star,\bm{\pi}_2^\star)$. 

Hence, in summary, it must hold that $\operatorname{supp}(\bm{\pi}_1^\star)\cap S_1^c=\emptyset$, i.e., $\operatorname{supp}(\bm{\pi}_1^\star)\subseteq S_1$.
\section{Impossibility of Preference Matching}
\label{sec:bias}

In this section, we first revisit the definition of preference matching in the BTL model \citep{xiao2024algorithmic}. Then we introduce a general theoretical framework of preference matching within the context of game-theoretic LLM alignment, and establish a general impossibility result, as stated in Theorem \ref{thm: no pm}. Finally we apply this general result to problem \eqref{eq:gen-nlhf}, concluding that preference matching is impossible. The proof of Theorem \ref{thm: no pm} is provided in Section \ref{sec:proof-no-pm}.

In previous sections, we have characterized the diversity of alignment result via mixed strategies. In \citet{xiao2024algorithmic}, the authors proposed a more refined criterion for diversity when the preference $\cP(y \succ  y' \mid  x)$ follows a BTL model, 
\begin{equation*}
    \cP(y \succ y'\mid x)=\frac{\exp(r(x,y))}{\exp(r(x,y))+\exp(r(x,y'))}\,,
\end{equation*}
as given by \eqref{eq:BT}. 
They pointed out that it is unwise to completely disregard any minority opinions in the case that 51\% of human labelers prefer $y_1$ over $y_2$ for a binary comparison. 
They suggested that the policy \eqref{eq: PM rlhf def},
\begin{equation}
\label{eq: PM rlhf def}
    \bm{\pi}^*(y\mid  x)=\frac{\exp{(r(x,y))}}{\sum_{y^\prime}\exp{(r(x,y^\prime))}}\,,
\end{equation}
referred to as the preference-matching policy, fully accounts for the diversity in human preferences.

It is easy to see that there exists a Condorcet winning response under BTL model. According to Theorem \ref{thm:condorcet-consistent}, using  preference $\Psi(\cP(y \succ y' \mid  x))$ as payoff with $\Psi(t)=t$ or $\Psi(t)=\log(t/(1-t))$ will lead the Nash solution to collapse to a single response instead of matching with $\bm{\pi}^*$. This shows that both \texttt{RLHF} and \texttt{NLHF} do not accounts for the diversity inherent in human preferences from the perspective of preference matching \citep{xiao2024algorithmic,liu2025statistical}, even under BTL model.

To achieve alignment fully accounting for diversity, we would like to match the Nash solution with the desired policy $\bm{\pi}^*$. 
In \citet{xiao2024algorithmic}, the authors answered this question for \texttt{RLHF}. They proposed the preference matching RLHF (\texttt{PM-RLHF}) method which successfully achieves preference matching, by slightly modifying the \texttt{RLHF} objective.
Here, we aim to explore the possibility of designing a new learnable payoff matrix that aligns with the desired strategy in a game-theoretic framework for LLM alignment:
\begin{center}
\emph{ Which choices of $\Psi$ ensure preference matching?}
\end{center}

Although it is currently unknown how to generalize the notion of preference matching policy to a general non-BTL preference, to maintain the generality of the discussion and drop the BTL model assumption, we suppose there exists an ideal policy, denoted by $\bm{\pi}^*$, which captures the diversity of human preferences perfectly.

Given a prompt $x$, we consider the set of all possible responses generated by the LLM: $\{y_1, \ldots, y_n\}$. We further suppose that the policy $\bm{\pi}^*$ has full support over these $n$ responses, meaning $\bm{\pi}^* > 0$, as we exclude responses not supported by $\bm{\pi}^*$ from consideration. Then our goal is to construct a game, represented by a payoff matrix $\{\alpha_{ij}\}_{i,j=1}^n$, with its Nash solution the given policy $\bm{\pi}^*$, i.e.,
\begin{equation*}
    \bm{\pi}^*=\arg\max_{\bm{\pi}} \min_{\bm{\pi}^{\prime}} \sum_{i=1}^n \sum_{j=1}^n \alpha_{ij}\pi_i\pi'_j\,.
\end{equation*}

To answer this question, we characterize the Nash solution under the given payoff matrix $\{\alpha_{ij}\}_{i,j=1}^n$, which is summarized by Lemma \ref{lem: kkt}. The proof is deferred to Appendix \ref{pf: kkt}.

\begin{lemma}[KKT Condition]\label{lem: kkt}
    Consider a game with payoff matrix $\{\alpha_{ij}\}_{i,j=1}^n$. Then $\bm{\pi}^*>0$ is a Nash solution to the game if and only if there exists $\bm{u}^*\in\mathbb{R}^n$ with $\bm{u}^*\geqslant 0$ and $\sum_{i=1}^n u_i^*=1$, and $t^*\in\mathbb{R}$ such that the following KKT conditions hold: \begin{equation*}
    \begin{cases}
        \sum_{i=1}^n \pi_i^*\alpha_{ij}-t^*\leqslant 0 & j=1,\cdots,n \\
        u_j^*\left(\sum_{i=1}^n\pi_i^*\alpha_{ij}-t^*\right)=0&j=1,\cdots,n\\
        \sum_{j=1}^n\alpha_{ij}u_j^*=t^*&i=1,\cdots,n
    \end{cases}\,.
\end{equation*}
\end{lemma}

According to Lemma \ref{lem: kkt}, it is easy to verify that the payoff matrix
\begin{equation}\label{eq:explicit construction I}
    \alpha_{ij}=\pi_i^*+\pi_j^*-\delta_{ij}\,,\forall\, 1\leqslant i,j\leqslant n\,,
\end{equation}
and the payoff matrix 
\begin{equation}\label{eq:explicit construction II}
    \alpha_{ij}=-\frac{\pi_j^*}{\pi_i^*}+n\delta_{ij}\,,\forall\, 1\leqslant i,j\leqslant n\,,
\end{equation}
both guarantee that $ \bm{\pi}^*$ is a Nash solution (the details are provided in Appendix \ref{app:verifying}). 
However, these payoff matrices do not depend on the given policy $\bm{\pi}^*$ in a reasonable way. The payoff matrix in Equation \eqref{eq:explicit construction I} is symmetric, making it difficult to interpret. Even worse, it depends on the raw value of \( \bm{\pi}^* \). In practice, \( \bm{\pi}^* \) is often only known up to a normalizing constant. For instance, the preference matching policy  \eqref{eq: PM rlhf def} includes a normalizing constant in the denominator that involves summing over \( n \) terms. This constant is hard to determine when \( n \) is large and unknown, as is often the case in LLMs. The payoff matrix in Equation \eqref{eq:explicit construction II} faces a similar issue as it explicitly depends on \( n \), which is an extremely large and unknown value in practice.

In summary, the above two payoff matrices rely on information that is often unavailable in practice, such as $n$ and the raw value of $\bm{\pi}^*$. 
What we can obtain in practice for the design of $\alpha_{ij}$ is the preference information between two responses $y_i$ and $y_j$, which we assume depends solely on the ratio between $\pi_i^*$ and $\pi_j^*$. 
When the preference satisfies the BTL model \eqref{eq:BT}, this assumption is justified by the fact that the preference between any two responses depends solely on the ratio of the values assigned by their corresponding preference matching policies \eqref{eq: PM rlhf def}.
From this practical consideration, we assume that the payoff matrix satisfies the following assumptions:
\begin{assumption}\label{ass:payoff}
    Given any $\bm{\pi}^*>0$, the payoff matrix \( \{\alpha_{ij}\}_{i,j=1}^n \) satisfies the following conditions:
    \begin{enumerate}
        \item For all \( i \in [n] \), \( \alpha_{ii} = C \) where \( C \) is a constant independent of \( \bm{\pi}^* \) and \( n \). In other words, the diagonal elements are the same constant.
        \item For all \( i,j \in [n] \) with $i\neq j$, \( \alpha_{ij} = f\left(\frac{\pi^*_i}{\pi^*_j}\right) \) for some smooth function \( f \) that is independent of \( \bm{\pi}^* \) and \( n \). In other words, the off-diagonal elements depend on the ratio \( \frac{\pi^*_i}{\pi^*_j} \) in the same way for all pairs \( (i, j) \) with $i\neq j$.
    \end{enumerate}
\end{assumption}

We emphasize that the above two assumptions are crucial for constructing a meaningful and practically learnable payoff matrix. Furthermore, for effective alignment, the payoff matrix should not only ensure \( \bm{\pi}^* \) to be a Nash solution, but \( \bm{\pi}^* \) must be the only Nash solution. The uniqueness requirement excludes trivial payoff matrices such as \( \alpha_{ij} = C \), where every \( \bm{\pi}^* > 0 \) is a Nash solution.

Unfortunately, in Theorem \ref{thm: no pm}, we prove that such a payoff matrix \( \{\alpha_{ij}\}_{i,j=1}^n \) does not exist generally. The proof can be found in Section \ref{sec:proof-no-pm}.

\begin{theorem}[Impossibility of Preference Matching for General Payoffs]
\label{thm: no pm}
    There does not exist a payoff matrix $\{\alpha_{ij}\}_{i,j=1}^n$ satisfying Assumption \ref{ass:payoff} such that for any given $\bm{\pi}^*>0$, the Nash solution to the game is unique and equals to $\bm{\pi}^*$. 
\end{theorem}
\begin{remark}
    If we relax Assumption \ref{ass:payoff} and allow the entries of the payoff matrix to depend on $n$, then the design \eqref{eq:explicit construction II} is actually eligible for preference matching.
\end{remark}

Theorem \ref{thm: no pm} implies that no simple mapping of the preference can yield a payoff that leads to preference matching.
As a special case of Theorem \ref{thm: no pm}, under BTL model, the generalized game in Equation \eqref{eq:gen-nlhf} with a smooth mapping $\Psi$,
\begin{equation*}
\max_{\bm{\pi}} \min_{\bm{\pi}^{\prime}} \mathbb{E}_{x\sim\rho} \left[\mathbb{E}_{y\sim \bm{\pi}(\cdot\mid x)}\mathbb{E}_{y'\sim \bm{\pi}^\prime(\cdot\mid x)}\left[\Psi\left(\cP(y\succ y^{\prime}\mid x)\right)\right]\right]\,,
\end{equation*}
cannot achieve preference matching. Therefore, we obtain the following corollary.

\begin{corollary}
    Problem \eqref{eq:gen-nlhf} with smooth mapping $\Psi$ cannot achieve preference matching. 
\end{corollary}

\subsection{Proof of Theorem \ref{thm: no pm}}
\label{sec:proof-no-pm}
We first present a useful lemma (Lemma \ref{lem: unique nash}) that further investigates the KKT conditions (Lemma \ref{lem: kkt}) when the payoff matrix induces a unique Nash equilibrium.
\begin{lemma}\label{lem: unique nash}
    If a game with the payoff matrix $\{\alpha_{ij}\}_{i,j=1}^n$ has a unique Nash solution $\bm{\pi}^*$, then for any $j\in[n]$, it must hold $u_j^*>0$ in the KKT conditions, and
    \[
    \sum_{i=1}^n\pi_i^*\alpha_{ij}=t^*.
    \]
\end{lemma}
\begin{proof}[Proof of Lemma \ref{lem: unique nash}]
    Suppose that the KKT conditions provide the unique Nash solution $(\bm{\pi}^*,\bm{u}^*,t^*)$. Then we define:
    \[
    \mathcal{J}_0:=\left\{j\in[n]: u_j^*\neq 0\right\},\text{ and } 
    \tilde{\mathcal{J}}_0:=\left\{j\in[n]: u_j^*=0\right\},
    \]
    with $\mathcal{J}_0\cup \tilde{\mathcal{J}}_0=[n]$.
    Since $\bm{u}^*\geqslant 0$ and $\sum u_j^*=1$, there exists $j\in [n]$, such that $u_j^*\neq 0$, i.e., $\mathcal{J}_0\neq \emptyset$. 
    Now, we aim to show $\tilde{\mathcal{J}}_0=\emptyset$. 
    We prove by contradiction.
    Suppose $\tilde{\mathcal{J}}_0\neq\emptyset$, taking $j_0\in\mathcal{J}_0$, we consider two spaces
    \[
    \begin{aligned}
        &\ V_1:=\left\{\bm{\pi}\in\mathbb{R}^n:\sum_{i=1}^n\pi_i\left(\alpha_{ij}-\alpha_{ij_0}\right)=0,\ \forall j\in \mathcal{J}_0\backslash \{j_0\} \right\},\\
        &\ V_2:=V_1\bigcap\left\{
        \bm{\pi}\in\mathbb{R}^n: \sum_{i=1}^n\pi_i(\alpha_{ij}-\alpha_{ij_0})\leqslant 0,\ \forall j\in\tilde{\mathcal{J}}_0
        \right\}.
    \end{aligned}
    \]
    Then we claim that $\bm{\pi}^*\in V_2$ and $\operatorname{dim}(V_2)\geqslant 2$.
    For the first claim, by the KKT conditions in Lemma \ref{lem: kkt}, for any $j\in\mathcal{J}_0$, we obtain
    \[
    \sum_{i=1}^n\pi^*_i\alpha_{ij}=t^*=\sum_{i=1}^n\pi^*_i\alpha_{ij_0},
    \]
    thus $\bm{\pi}^*\in V_1$. 
    Moreover, again by the KKT conditions, for any $j\in\tilde{\mathcal{J}}_0$, we have 
    \[
    \sum_{i=1}^n\pi_i^*\alpha_{ij}\leqslant t^*= \sum_{i=1}^n\pi^*_i\alpha_{ij_0},
    \]
    which shows that $\bm{\pi}^*\in V_2$.
    For the second claim, take $\tilde{j}_0\in\tilde{\mathcal{J}}_0$ and consider 
    \[
    \begin{aligned}
        &\ V_3:=\left\{\bm{\pi}\in\mathbb{R}^n:\sum_{i=1}^n\pi_i\left(\alpha_{ij}-\alpha_{ij_0}\right)=0,\ \forall j\in [n]\backslash \{j_0,\tilde{j}_0\}\right\},\\
        &\ V_4:=V_3\bigcap\left\{
        \bm{\pi}\in\mathbb{R}^n: \sum_{i=1}^n\pi_i(\alpha_{i\tilde{j}_0}-\alpha_{ij_0})\leqslant 0
        \right\}.
    \end{aligned}
    \]
    We can easily see $V_4\subseteq V_2$. 
    Note that $V_3$ can be regarded as a kernel space of a linear transformation from $\mathbb{R}^n$ to $\mathbb{R}^{n-2}$. 
    By the dimension theorem in linear algebra, we obtain $\operatorname{dim}(V_3)=n-\operatorname{dim}(\operatorname{Im}(A))\geqslant n-(n-2)=2$. 
    For any $\bm{\pi}\in V_3$, it must hold that $\bm{\pi}\in V_4$ or $-\bm{\pi}\in V_4$, so $\operatorname{dim}(V_4)=\operatorname{dim}(V_3)\geqslant 2$. 
    Therefore, we have $\operatorname{dim}(V_1)\geqslant\operatorname{dim}(V_2)\geqslant\operatorname{dim}(V_4)\geqslant 2$.

    Thus, we can take another $\tilde{\bm{\pi}}^*\in V_2$ which is linear independent with $\bm{\pi}^*$. 
    Note that for any $a,b\in\mathbb{R}_+$, $a\bm{\pi}^*+b\tilde{\bm{\pi}}^*\in V_2$. 
    Taking large $a\in\mathbb{R}_+$, we have $a\bm{\pi}^*+b\tilde{\bm{\pi}}^*\in V_2$ and $a\bm{\pi}^*+b\tilde{\bm{\pi}}^*>0$, since $\pi^*>0$. 
    Therefore, there exists $a_1\in\mathbb{R}_+$, such that $\bm{\pi}_2^*:=\frac{a\bm{\pi}^*+\tilde{\bm{\pi}}^*}{a_1}\in V_2$ that satisfies $\bm{\pi}_2^*\neq \bm{\pi}^*$, $\bm{\pi}_2^*>0$, and $\sum_i \pi_{2,i}^*=1$. 
    Thus, we obtain another Nash equilibrium $(\bm{\pi}_2^*, \bm{u}^*, t^*)$, causing contradiction to the uniqueness of Nash solution. 
    Hence, it must hold that $\tilde{\mathcal{J}}_0=\emptyset$.
\end{proof}

Next we provide the proof of Theorem \ref{thm: no pm}. Using Lemma \ref{lem: unique nash}, uniqueness requires us to seek solutions that satisfies $$\sum_{i=1}^n\pi^*_i\alpha_{ij}=t^*$$ for all $j\in [n]$, where $t^*$ is a constant that may depend on $\bm{\pi}^*$.
Consider $n\geqslant 5$, for any four distinct indices $j_1, j_2, k_1, k_2$, we have
\begin{equation}\label{eq: equation before variation}
    \begin{aligned}
&\ \sum_{i\neq j_1,k_1,k_2}\pi_if\left(
\frac{\pi_i}{\pi_{j_1}}\right) + \pi_{k_1}f\left(\frac{\pi_{k_1}}{\pi_{j_1}}
\right) + \pi_{k_2}f\left(\frac{\pi_{k_2}}{\pi_{j_1}}\right) + C\pi_{j_1}\\
= &\ \sum_{i\neq j_2,k_1,k_2}\pi_if\left(
\frac{\pi_i}{\pi_{j_2}}\right) + \pi_{k_1}f\left(\frac{\pi_{k_1}}{\pi_{j_2}}
\right) + \pi_{k_2}f\left(\frac{\pi_{k_2}}{\pi_{j_2}}\right) + C\pi_{j_2}
\end{aligned}
\end{equation}
Let us consider the infinitesimal variation $\pi_{k_1}\to \pi_{k_1}+\delta$ and $\pi_{k_2}\to \pi_{k_2}-\delta$, keeping others still. We obtain that 
\begin{equation}\label{eq: equation after variation}
    \begin{aligned}
&\ \sum_{i\neq j_1,k_1,k_2}\pi_if\left(
\frac{\pi_i}{\pi_{j_1}}\right) + (\pi_{k_1}+\delta)f\left(\frac{\pi_{k_1}+\delta}{\pi_{j_1}}
\right) + (\pi_{k_2}-\delta)f\left(\frac{\pi_{k_2}-\delta}{\pi_{j_1}}\right) + C\pi_{j_1}\\
= &\ \sum_{i\neq j_2,k_1,k_2}\pi_if\left(
\frac{\pi_i}{\pi_{j_2}}\right) + (\pi_{k_1}+\delta)f\left(\frac{\pi_{k_1}+\delta}{\pi_{j_2}}
\right) + (\pi_{k_2}-\delta)f\left(\frac{\pi_{k_2}-\delta}{\pi_{j_2}}\right) + C\pi_{j_2}
\end{aligned}
\end{equation}
Subtracting both sides of \eqref{eq: equation before variation} from \eqref{eq: equation after variation}, we obtain that 
\begin{equation}\label{eq:equation variation difference}
    \begin{aligned}
&\ (\pi_{k_1}+\delta)f\left(\frac{\pi_{k_1}+\delta}{\pi_{j_1}}
\right) + (\pi_{k_2}-\delta)f\left(\frac{\pi_{k_2}-\delta}{\pi_{j_1}}\right) - \pi_{k_1}f\left(\frac{\pi_{k_1}}{\pi_{j_1}}
\right) - \pi_{k_2}f\left(\frac{\pi_{k_2}}{\pi_{j_1}}\right) \\
= &\  
(\pi_{k_1}+\delta)f\left(\frac{\pi_{k_1}+\delta}{\pi_{j_2}}
\right) + (\pi_{k_2}-\delta)f\left(\frac{\pi_{k_2}-\delta}{\pi_{j_2}}\right) -
\pi_{k_1}f\left(\frac{\pi_{k_1}}{\pi_{j_2}}
\right) - \pi_{k_2}f\left(\frac{\pi_{k_2}}{\pi_{j_2}}\right)\,,
    \end{aligned}
\end{equation}
i.e., we have 
\begin{equation}\label{eq: equation delta}
    \begin{aligned}
        &\ (\pi_{k_1}+\delta)\left(f\left(\frac{\pi_{k_1}+\delta}{\pi_{j_1}}
\right)-f\left(\frac{\pi_{k_1}}{\pi_{j_1}}
\right)\right) + \delta f\left(\frac{\pi_{k_1}}{\pi_{j_1}}
\right) \\
&\ \ + (\pi_{k_2}-\delta)\left(f\left(\frac{\pi_{k_2}-\delta}{\pi_{j_1}}\right) - f\left(\frac{\pi_{k_2}}{\pi_{j_1}}\right)\right) - \delta f\left(\frac{\pi_{k_2}}{\pi_{j_1}}\right)\\
= &\ (\pi_{k_1}+\delta)\left(f\left(\frac{\pi_{k_1}+\delta}{\pi_{j_2}}
\right)-f\left(\frac{\pi_{k_1}}{\pi_{j_2}}
\right)\right) + \delta f\left(\frac{\pi_{k_1}}{\pi_{j_2}}
\right) \\
&\ \ + (\pi_{k_2}-\delta)\left(f\left(\frac{\pi_{k_2}-\delta}{\pi_{j_2}}\right) - f\left(\frac{\pi_{k_2}}{\pi_{j_2}}\right)\right) - \delta f\left(\frac{\pi_{k_2}}{\pi_{j_2}}\right).
    \end{aligned}
\end{equation}
As $f$ is smooth, using  
\[
\lim_{\delta\to 0}\frac{f\left(\frac{x+\delta}{\pi_j}\right)-f\left(\frac{x}{\pi_j}\right)}{\delta}=\frac{1}{\pi_j}f^\prime\left(\frac{x}{\pi_j}\right),
\]
and taking $\delta\to 0$, we obtain the following identity from \eqref{eq: equation delta},
\begin{equation}
    \begin{aligned}
        &\ f\left(\frac{\pi_{k_1}}{\pi_{j_1}}\right) + \frac{\pi_{k_1}}{\pi_{j_1}}f^\prime\left(\frac{\pi_{k_1}}{\pi_{j_1}}\right) - f\left(\frac{\pi_{k_2}}{\pi_{j_1}}\right) - \frac{\pi_{k_2}}{\pi_{j_1}}f^\prime\left(\frac{\pi_{k_2}}{\pi_{j_1}}\right)\\
        = &\ f\left(\frac{\pi_{k_1}}{\pi_{j_2}}\right) + \frac{\pi_{k_1}}{\pi_{j_2}}f^\prime\left(\frac{\pi_{k_1}}{\pi_{j_2}}\right) - f\left(\frac{\pi_{k_2}}{\pi_{j_2}}\right) - \frac{\pi_{k_2}}{\pi_{j_2}}f^\prime\left(\frac{\pi_{k_2}}{\pi_{j_2}}\right).\\
    \end{aligned}
\end{equation}
Thus, we obtain that 
\begin{equation}\label{eq: equation align 1}
    \begin{aligned}
        &\ f\left(\frac{\pi_{k_1}}{\pi_{j_1}}\right) + \frac{\pi_{k_1}}{\pi_{j_1}}f^\prime\left(\frac{\pi_{k_1}}{\pi_{j_1}}\right) -f\left(\frac{\pi_{k_1}}{\pi_{j_2}}\right) -\frac{\pi_{k_1}}{\pi_{j_2}}f^\prime\left(\frac{\pi_{k_1}}{\pi_{j_2}}\right) \\
        = &\ f\left(\frac{\pi_{k_2}}{\pi_{j_1}}\right) + \frac{\pi_{k_2}}{\pi_{j_1}}f^\prime\left(\frac{\pi_{k_2}}{\pi_{j_1}}\right) - f\left(\frac{\pi_{k_2}}{\pi_{j_2}}\right) - \frac{\pi_{k_2}}{\pi_{j_2}}f^\prime\left(\frac{\pi_{k_2}}{\pi_{j_2}}\right).\\
    \end{aligned}
\end{equation}
Since \eqref{eq: equation align 1} holds for any $\pi>0$, given any $\pi_{j_1}\neq\pi_{j_2}$, for any $x_1,x_2\in(0,1-\pi_{j_1}-\pi_{j_2})$, we have
\[
\begin{aligned}
        &\ f\left(\frac{x_1}{\pi_{j_1}}\right) + \frac{x_1}{\pi_{j_1}}f^\prime\left(\frac{x_1}{\pi_{j_1}}\right) -f\left(\frac{x_1}{\pi_{j_2}}\right) -\frac{x_1}{\pi_{j_2}}f^\prime\left(\frac{x_1}{\pi_{j_2}}\right) \\
        = &\ f\left(\frac{x_2}{\pi_{j_1}}\right) + \frac{x_2}{\pi_{j_1}}f^\prime\left(\frac{x_2}{\pi_{j_1}}\right) - f\left(\frac{x_2}{\pi_{j_2}}\right) - \frac{x_2}{\pi_{j_2}}f^\prime\left(\frac{x_2}{\pi_{j_2}}\right),\\
    \end{aligned}
\]
which induces the following for any $x\in (0,1-\pi_{j_1}-\pi_{j_2})$, 
\begin{equation}
f\left(\frac{x}{\pi_{j_1}}\right) + \frac{x}{\pi_{j_1}}f^\prime\left(\frac{x}{\pi_{j_1}}\right) -f\left(\frac{x}{\pi_{j_2}}\right) -\frac{x}{\pi_{j_2}}f^\prime\left(\frac{x}{\pi_{j_2}}\right) = C(\pi_{j_1},\pi_{j_2})\,,
\end{equation}
i.e., we have
\begin{equation}
    f\left(\frac{x}{\pi_{j_1}}\right) + \frac{x}{\pi_{j_1}}f^\prime\left(\frac{x}{\pi_{j_1}}\right) = C(\pi_{j_1},\pi_{j_2}) + f\left(\frac{x}{\pi_{j_2}}\right) + \frac{x}{\pi_{j_2}}f^\prime\left(\frac{x}{\pi_{j_2}}\right).
\end{equation}
Without any loss of generality, we assume $\pi_{j_1}<\pi_{j_2}$, then we obtain 
\[
\begin{aligned}
    &\ f\left(\frac{x}{\pi_{j_1}}\right) + \frac{x}{\pi_{j_1}}f^\prime\left(\frac{x}{\pi_{j_1}}\right)\\
    = &\ C(\pi_{j_1},\pi_{j_2}) + f\left(\frac{x}{\pi_{j_2}}\right) + \frac{x}{\pi_{j_2}}f^\prime\left(\frac{x}{\pi_{j_2}}\right)\\
    = &\ 2C(\pi_{j_1},\pi_{j_2}) + f\left(\frac{\pi_{j_1}x}{\pi^2_{j_2}}\right) + \frac{\pi_{j_1}x}{\pi^2_{j_2}}f^\prime\left(\frac{\pi_{j_1}x}{\pi^2_{j_2}}\right)\\
    = &\ \cdots\cdots \\
    = &\ nC(\pi_{j_1},\pi_{j_2}) + f\left(\frac{\pi^{n-1}_{j_1}x}{\pi^n_{j_2}}\right) + \frac{\pi^{n-1}_{j_1}x}{\pi^n_{j_2}}f^\prime\left(\frac{\pi^{n-1}_{j_1}x}{\pi^n_{j_2}}\right)\\
    = &\ \cdots\cdots.
\end{aligned}
\]
Taking limit, it must hold $C(\pi_{j_1},\pi_{j_2})=0$, i.e.,  we have
\begin{equation}\label{eq: equation align 2}
    f\left(\frac{x}{\pi_{j_1}}\right) + \frac{x}{\pi_{j_1}}f^\prime\left(\frac{x}{\pi_{j_1}}\right) = f\left(\frac{x}{\pi_{j_2}}\right) + \frac{x}{\pi_{j_2}}f^\prime\left(\frac{x}{\pi_{j_2}}\right).
\end{equation}
Since \eqref{eq: equation align 2} holds for any $\pi>0$, for any $x_1,x_2\in\mathbb{R}_+$, we have
\[
f(x_1)+x_1f^\prime(x_1)=f^\prime(x_2)+x_2f^\prime(x_2),
\]
thus, for any $x\in\mathbb{R}_+$, we have
\begin{equation}\label{eq: no pm ode}
    f\left(x\right) + xf^\prime\left(x\right) = C_1\,.
\end{equation}
Solving \eqref{eq: no pm ode}, we obtain that
$$
f(x)=\frac{C_2}{x}+C_3\,.
$$
Then we obtain that 
\[
\sum_{i=1}^n\pi_i\alpha_{ij}=C\pi_j + \sum_{i\neq j}\pi_i\left(\frac{C_2\pi_j}{\pi_i}+C_3\right)
= C_3 + (C+(n-1)C_2-C_3)\pi_j\,,
\]
yielding $C_3=C+(n-1)C_2$, and 
\[
f(x)=C + C_2\left(\frac{1}{x}+n-1\right),
\]
which is contradictory to our assumptions.

\section{Conclusions}
\label{sec:discuss}

We have investigated several properties motivated by social choice theory and diversity considerations within the general game-theoretic LLM alignment framework \eqref{eq:gen-nlhf}, where the payoff is designed as a mapping $\Psi$ of the original preference. We have identified the necessary and sufficient conditions on $\Psi$ to guarantee Condorcet consistency and Smith consistency. These conditions allow for a considerably broad class of choices for $\Psi$, demonstrating that these desirable alignment properties are not sensitive to the exact values of the payoff, thereby providing a theoretical foundation for the robustness of the game-theoretic LLM alignment approach. Additionally, we have examined conditions on $\Psi$ that ensure the resulting policy is a mixed strategy, preserving diversity in human preferences. Finally, we have proved that achieving exact preference matching is impossible under the general game-theoretic alignment framework with a smooth mapping, revealing fundamental limitations of this approach.

Our findings suggest several promising directions for future research on LLM alignment. First, while we establish an impossibility result for preference matching under the assumption that $\Psi$ is smooth, it remains an open question whether preference matching can be achieved when $\Psi$ is merely continuous. Second, in practical settings, regularization terms based on the reference model are often added to problem \eqref{eq:gen-nlhf}. Regularization may be crucial for preference matching, for example, \citet{xiao2024algorithmic} modified the regularization term in \texttt{RLHF} to achieve preference matching. Analyzing the alignment properties of game-theoretic methods with such regularization is another interesting avenue for future work. Furthermore, how to explicitly define a preference-matching policy for general preferences that do not satisfy the BTL model, and how to develop alignment approaches capable of learning such a policy, remain open problems. Finally, our results highlight that practical preference models must satisfy certain anti-symmetry conditions to ensure Smith consistency---conditions that are not guaranteed by several currently used models. Thus, designing preference model architectures that enforce anti-symmetry is an important and interesting future direction.

{\small
\section*{Acknowledgments} 
We are grateful to Yuhe Qin for providing helpful discussions.
This work was supported in part by NIH grant U01CA274576, ARPA-H Award D24AC00253, NSF grant DMS-2310679, a Meta Faculty Research Award, and Wharton AI for Business.

\bibliographystyle{abbrvnat}
\bibliography{Bibliography}
}

\clearpage
\appendix
\section{Proof of Lemma \ref{lem: kkt}}\label{pf: kkt}
Suppose each player has $n$ policies and the payoff matrix is $\{\alpha_{ij}\}_{i=1}^n$. Then,
$$
\begin{aligned}
\max_{\bm{\pi}}\min_{\bm{\pi}^{\prime}}
\left\{
\sum_{i=1}^n\sum_{j=1}^n\alpha_{ij}\pi_i\pi^{\prime}_j\right\}
 =& \max_{\bm{\pi}}\min_{\bm{\pi}^{\prime}}
\left\{
\sum_{j=1}^n \left(\sum_{i=1}^n\alpha_{ij}\pi_i\right)\pi^{\prime}_j\right\}
= \max_{\bm{\pi}}\min_j
\left\{
\sum_{i=1}^n\alpha_{ij}\pi_i
\right\} .
\end{aligned}
$$
Let us reformulated it into a convex optimization problem.
\begin{equation}\tag{$P$}
    \begin{aligned}
\min_{\bm{\pi}}\quad\quad\quad &\max_j \sum_{i=1}^n\alpha_{ij}\pi_i \\
\text{subject to} \ \
& -\pi_i \leqslant 0, \quad i = 1, \ldots, n \\
& \sum_{i=1}^n \pi_i- 1=0
\end{aligned}
\end{equation}
Let us further reformulate this problem into the epigraph form by introducing a single variable $t\in\mathbb{R}$:
\begin{equation}\tag{$P^\prime$}
    \begin{aligned}
\min_{\bm{\pi},t}\quad\quad\quad &t \\
\text{subject to}  \ \ 
& \sum_{i=1}^n\alpha_{ij}\pi_i-t\leqslant 0, \quad j = 1, \ldots, n \\
& -\pi_i \leqslant 0, \quad i = 1, \ldots, n \\
& \sum_{i=1}^n \pi_i- 1=0
\end{aligned}
\end{equation}
By introducing the dual variables $\bm{u}^*\in\mathbb{R}^n$, $\tilde{\bm{u}}^*\in\mathbb{R}^n$ and $v^*\in\mathbb{R}$, the KKT conditions is:
\begin{itemize}
    \item stationary condition: $$\sum_{j=1}^n\alpha_{ij}u_j^*-\tilde{u}_i^*=-v^* \quad i=1,\cdots,n$$
    \item complementary slackness: $$u_j^*\left(\sum_{i=1}^n\pi_i^*\alpha_{ij}-t^*\right)=0\quad j=1,\cdots,n$$
    $$\tilde{u}^*_i\pi^*_i=0\quad i=1,\cdots,n$$
    \item primal feasibility:
    $$\sum_{i=1}^n \pi_i^*\alpha_{ij}-t^*\leqslant 0$$
    $$\bm{\pi}^* \geqslant 0$$
    $$\sum_{i=1}^n \pi_i^*=1$$
    \item dual feasibility:
    $$\bm{u}^*\geqslant 0$$
    $$\sum_{i=1}^nu^*_i=1$$
    $$\tilde{\bm{u}}^*\geqslant 0$$
\end{itemize}

We can easily see that Slater's condition is satisfied for this problem, so the KKT points are equivalent to primal and dual solutions.
Then taking $\bm{\pi}^*>0$ into account, we have $\tilde{u}_i^*=0$ by the second complementary slackness condition, and the above equations can be simplified to the following system of equations:
\begin{equation*}
    \begin{cases}
        \bm{u}^*\geqslant 0 & \\
        \sum_{i=1}^nu^*_i=1 & \\
        \sum_{i=1}^n \pi_i^*\alpha_{ij}-t^*\leqslant 0 & j=1,\cdots,n \\
        u_j^*\left(\sum_{i=1}^n\pi_i^*\alpha_{ij}-t^*\right)=0&j=1,\cdots,n\\
        \sum_{j=1}^n\alpha_{ij}u_j^*=-v^*&i=1,\cdots,n
    \end{cases}
\end{equation*}
Moreover, notice that 
\[
0=\sum_{j=1}^nu_j^*\left(\sum_{i=1}^n\pi_i^*\alpha_{ij}-t^*\right)=\sum_{i=1}^n\pi_i^*\sum_{j=1}^n\alpha_{ij}u_j^* - t^* =-v^*-t^*\,,
\]
thus $v^*=-t^*$. Hence, we conclude our proof.

\section{Verifying Equation \texorpdfstring{\eqref{eq:explicit construction I}}{(I)} and Equation \texorpdfstring{\eqref{eq:explicit construction II}}{(II)}}
\label{app:verifying}
For \eqref{eq:explicit construction I}, choosing $t^*=-v^*=\sum_{i=1}^n (\pi^*_i)^2$ and $u_i^*=v_i^*$, $\bm{\pi}^*$ is a Nash solution. 
For \eqref{eq:explicit construction II}, choosing $t^*=-v^*=0$ and $u_j^*=\frac{(\pi_j^*)^{-1}}{\sum_{j=1}^n(\pi_j^*)^{-1}}$, $\bm{\pi}^*$ is a Nash solution.

\end{document}